\documentclass[english]{article}
\usepackage[T1]{fontenc}
\usepackage[latin9]{inputenc}
\PassOptionsToPackage{natbib=true}{biblatex}
\usepackage{xcolor}
\usepackage{pdfcolmk}
\usepackage{array}
\usepackage{rotating}
\usepackage{booktabs}
\usepackage{mathtools}
\usepackage{amsmath}
\usepackage{amsthm}
\usepackage{amssymb}
\usepackage{graphicx}
\PassOptionsToPackage{normalem}{ulem}
\usepackage{ulem}

\makeatletter

\providecommand{\tabularnewline}{\\}
\providecolor{lyxadded}{rgb}{0,0,1}
\providecolor{lyxdeleted}{rgb}{1,0,0}

\DeclareRobustCommand{\lyxsout}[1]{\ifx\\#1\else\sout{#1}\fi}

\usepackage{thmtools}
\usepackage{thm-restate}

\@ifundefined{showcaptionsetup}{}{%
 \PassOptionsToPackage{caption=false}{subfig}}
\usepackage{subfig}
\makeatother

\usepackage{babel}
\usepackage[style=authoryear]{biblatex}
\begin{document}
\global\long\def\P{\mathbf{P}}%
\global\long\def\E{\mathbf{E}}%
\global\long\def\Var{\mathrm{Var}}%
\global\long\def\transpose{^{\mathsf{T}}}%
\global\long\def\mc#1{\mathcal{#1}}%
\global\long\def\diag{\mathrm{diag}}%

\title{Optimal Multivariate Tuning with Neuron-Level and Population-Level
Energy Constraints}
\author{Yuval Harel, Ron Meir\\
Department of Electrical Engineering,\\
Technion -- Israel Institute of Technology, Haifa, Israel}
\maketitle
\begin{abstract}
Optimality principles have been useful in explaining many aspects
of biological systems. In the context of neural encoding in sensory
areas, optimality is naturally formulated in a Bayesian setting, as
neural tuning which minimizes mean decoding error. Many works optimize
Fisher information, which approximates the Minimum Mean Square Error
(MMSE) of the optimal decoder for long encoding time, but may be misleading
for short encoding times. We study MMSE-optimal neural encoding of
a multivariate stimulus by uniform populations of spiking neurons,
under firing rate constraints for each neuron as well as for the entire
population. We show that the population-level constraint is essential
for the formulation of a well-posed problem having finite optimal
tuning widths, and optimal tuning aligns with the principal components
of the prior distribution. Numerical evaluation of the two-dimensional
case shows that encoding only the dimension with higher variance is
optimal for short encoding times. We also compare direct MMSE optimization
to optimization of several proxies to MMSE, namely Fisher information,
Maximum Likelihood estimation error, and the Bayesian Cramér-Rao bound.
We find that optimization of these measures yield qualitatively misleading
results regarding MMSE-optimal tuning and its dependence on encoding
time and energy constraints.
\end{abstract}

\section{Introduction}

Optimality principles have been a useful tool in explaining many aspects
of biological systems, including motor behavior \citep{Todorov2004Optimality},
perception \citep{Gardner2019}, and neural activity \citep{Berkes2011Spontaneous}.
In the context of neural encoding in sensory areas, optimality is
naturally formulated in a Bayesian setting, by specifying \emph{(i)}
a prior distribution of encoded stimuli; \emph{(ii)} a performance
criterion, such as mean decoding error or motor performance; \emph{(iii)}
optimization constraints, such as constraints on energy consumption.
Empirically, neural encoding has been found to depend on the statistics
of natural stimuli \citep{Berkes2011Spontaneous,HarMcAlp04}, as well
as to rapidly adapt to the statistics of recently presented stimuli
\citep{Dean2008Rapid,Benucci2009}, indicating the importance of the
prior distribution. Sensory adaptation also occurs following motor
learning \citep{Darainy2018}, suggesting the relevance of motor performance
criteria, though the effect of motor learning on sensory perception
is quite small. Since sensory encoding occurs in the context of an
acting organism, an optimality criterion should ideally take the entire
sensory-motor loop into account. However, to make the analysis tractable,
most existing work has focused on task-independent measures of encoding
performance, such as decoding error.

A natural optimality criterion for neural encoding is the Mean Square
Error (MSE) of a subsequent optimal decoder, namely the Minimal Mean
Square Error (MMSE) estimator. Many works optimize Fisher information,
which serves as a proxy to the MSE (see \citep{Pilarski2015} for
a thorough review). Specifically, the Cramér-Rao Bound (CRB) states
that the MSE of \emph{unbiased} estimators is lower-bounded by the
inverse of Fisher information. More importantly, under appropriate
regularity conditions, the optimal MSE approaches the expected value
of the CRB \citep{VanTrees2004Detection} in the asymptotic regime
of low noise, corresponding to large decoding time or large population
firing rate. Similar asymptotic relations can be derived between Fisher
information and general $L_{p}$ estimation errors \citep{WanStoLee2016}.
Fisher information of neural spiking activity is easy to compute analytically,
at least in the case of static stimulus \citep[section 3.3]{DayAbb05},
and it can be used without characterizing the decoder. However, optimizing
Fisher information may yield misleading qualitative results regarding
the MSE-optimal encoding outside of the asymptotic regime of large
decoding time \citep{Bethge2002,YaeMei10,Pilarski2015}. Ideally,
a model of biological encoding as optimal encoding should rely on
biologically relevant optimality criteria such as decoding error,
rather than a proxy to decoding error such as Fisher information.
In particular, outside of the asymptotic regime, the CRB does not
provide adequate justification for the use of Fisher information,
since optimal decoding is typically biased, so the CRB does not apply
to the optimal decoder.

Several previous works have attempted direct minimization of decoding
MSE rather than Fisher information. \citep{bethge_optimal_2003} and
two sequel articles \citep{bethge_binary_2003,bethge_second_2003}
analytically derive decoding MMSE for encoding a static scalar state
uniformly distributed on $\left[0,1\right]$ by a single neuron, for
a wide class of piecewise-power-law tuning functions. They find a
phase transition between binary encoding, which is optimal for short
coding times, and analog encoding, which is optimal for long coding
times. In \citep{YaeMei10}, an explicit expression for the MSE of
the optimal Bayesian decoder is derived for a static state encoded
by a uniform population of Gaussian neurons, and is used to characterize
optimal tuning function width and its relation to coding time in the
encoding of scalar stimuli. More recently, \citep{WanStoLee2016}
studied $L_{p}$-based loss measures in the asymptotic regime using
Fisher information, as well as numerically for the maximum a-posteriori
estimator outside the asymptotic regime. \citet{Finkelstein2018}
study the MSE of the ML estimator in the context of two-dimensional
encoding of angles, used as a model for head direction encoding in
bats. In \citep{SusMeiOpp11,SusMeiOpp13}, a mean-field approximation
is suggested to allow efficient evaluation of the MSE in a dynamic
setting. We have previously used Assumed Density Filtering for the
approximate evaluation of decoding MMSE for non-uniform populations
and dynamic state in \citep{Harel2018}.

As far as we are aware, most previous works studying optimal neural
encoding have paid little attention to the topic of energy constraints
or other optimization constraints. Most commonly, optimization is
performed over tuning width while keeping the maximal firing rate
of each neuron fixed. In some cases (such as \citep{WanStoLee2016}),
this is explicitly stated as a constraint on each neuron's firing
rate. However, even when the maximal firing rate of each neuron is
constrained, tuning width affects the mean total firing rate -- and
therefore energy consumption -- of the entire neural population.
This issue has been addressed heuristically in \citep{Zhang1999}
by studying Fisher information per spike as a measure of energy efficiency.
However, the interpretation or significance of this ratio remains
unclear. \citep{Bethge2002} have addressed this issue more systematically
in the context of a univariate state taking values in a bounded interval,
with a constraint on the mean firing rate as well as maximal firing
rate of each neuron (see section 7 therein). MMSE was evaluated using
Monte-Carlo methods, showing that imposing the energy constraint leads
to narrower tuning. \citet{Ganguli2011} also study optimal encoding
with a parameterized population of neurons under a population rate
constraint. However their analysis is based on optimization of Fisher
information.

In this work, we focus on optimization of tuning widths in a multivariate
setting, with constraints on the firing rate of each neuron as well
as the entire population. To allow exact closed-form evaluation of
decoding MMSE, we focus on uniform Gaussian encoding of a static stimulus.
Optimization of tuning width involves a trade-off between population
firing rate, which is maximized by wide tuning, and selectivity of
individual neurons, which is maximized by narrow tuning \citep{Zhang1999,Eurich2000,Sun2017}.
Estimation MSE in a uniform Gaussian population is minimized for a
finite nonzero width in the univariate case \citep{YaeMei10}, and
for infinitely wide tuning in the multivariate case, as shown in section
\ref{sec:Optimal-Encoding}. The addition of a constraint on the population
firing rate gives rise to a finite-width solution in the multivariate
case, in which both constraints are active. The constraint on the
population firing rate may be interpreted as a sparsity constraint:
for a fixed maximal firing rate per neuron, a smaller population rate
constraint means less neurons may fire near their maximal rate for
each stimulus. Optimal tuning depends on both energy constraints,
as well as encoding time.

Applying the predictions of Bayesian MMSE-optimal tuning to explain
experimentally observed coding is challenging. A substantial difficulty
is the characterization of the relevant prior distribution that encoding
is adapted to. In a few special cases, this prior may be derived from
simple physical considerations (e.g. \citet{HarMcAlp04}), but such
an approach is not usually applicable. Another difficulty is in the
evaluation of decoding MMSE outside of the special case of Gaussian
uniform coding, or conversely, identifying biological circumstances
where the Gaussian uniform assumption is reasonable. Some biological
implications of MMSE-optimality in the scalar case have been related
to experimental results in \citep{YaeMei10}. In the multivariate
case, there are fewer relevant experimental results where multivariate
tuning is measured and the prior distribution and encoding time may
be determined. Possibly relevant results are in the encoding of head
direction in bats studied in \citep{Finkelstein2018}. However, it
is not clear that the assumption of Gaussian uniform coding is justified
in this case, as briefly discussed in section \ref{subsec:Two-dimensional-uniform-Gaussian}.

\textbf{Main results}: \emph{(i) }Closed-form characterization of
MMSE in multivariate uniform Gaussian encoding, as well as simple
lower and upper bounds. \emph{(ii)} A population-level constraint
is essential to formulate a well-posed optimization problem in a multivariate
setting, and optimal tuning depends on this constraint in a non-trivial
way.\emph{ (iii) }MMSE-optimal multivariate Gaussian tuning functions
are aligned with the principal components of the prior distribution
and tuning is narrower in dimensions with larger prior variance. \emph{(iv)}
The trade-off between firing rate and neuron selectivity shifts with
stimulus dimensionality, so that MSE-optimal multivariate encoding
always involves maximizing population firing rate, in contrast to
the univariate case. \emph{(v)} Optimization of proxies to MMSE such
as Fisher information or MSE of the Maximum Likelihood estimator yield
qualitatively misleading results. \emph{(vi)} In a two-dimensional
setting, one-dimensional encoding is preferred for short decoding
times, in contrast to previous predictions.

\section{Setting and notation}

\subsection{Optimal encoding and decoding}

We consider the problem of optimal encoding and decoding of a static
external state, described by a random $m$-dimensional variable $X\in\mathbb{R}^{m}$.
The state is observed through a population of sensory neurons over
a time interval $[0,T]$. Given the state, neurons fire independently,
with the $i$th neuron firing as a Poisson process with rate $\lambda^{\left(i\right)}\left(X\right)$.
We assume the set $\{\lambda^{(i)}\}_{i=1}^{M}$ belongs to some parameterized
family with parameter $\phi$, and denote the number of spikes from
the $i$th neuron up to time $T$ by $N_{T}^{\left(i\right)}$ and
the spike counts from all neurons by $\boldsymbol{N}_{T}\coloneqq(N_{T}^{\left(i\right)})_{i=1}^{M}$.
In this context, e\emph{ncoding }refers to the choice of functions
$\lambda^{\left(i\right)}$, while \emph{decoding} means estimating
the state $X$ from the spike counts\footnote{More generally, we could estimate the state $X$ from the exact spike
pattern rather than just spike count. However, since the estimated
state is static, i.e., does not vary through the observation interval
$\left[0,T\right]$, and the firing is Poisson, spike counts comprise
a sufficient static for this estimation problem, simplifying the
analysis.} $\boldsymbol{N}_{T}$.

Given an estimator $\hat{X}=\hat{X}\left(\boldsymbol{N}_{T}\right)$,
we define the Mean Square Error (MSE) as 
\begin{equation}
\epsilon(\hat{X})\coloneqq\mathrm{tr}[(X-\hat{X})(X-\hat{X})\transpose]\label{eq:mse}
\end{equation}
where $\mathrm{tr}$ is the trace operator. We seek an estimator (decoder)
$\hat{X}$ and tuning parameters (encoder) $\phi$ that solve 
\[
\min_{\phi}\min_{\hat{X}}\mathrm{E}\left[\epsilon\right]=\min_{\phi}\mathrm{E}[\min_{\hat{X}}\mathrm{E}[\epsilon|\boldsymbol{N}_{T}]].
\]
The inner minimization problem in this equation is solved by the MSE-optimal
decoder, which is the posterior mean $\hat{X}=\mu_{T}\coloneqq\mathrm{E}\left[X|\boldsymbol{N}_{T}\right]$.
We denote the error of this estimator by $\epsilon_{\mathrm{MMSE}}\coloneqq\epsilon\left(\mu_{T}\right)$
(for Minimum Mean Square Error). The outer minimization problem becomes
$\min_{\phi}\E\left[\epsilon_{\mathrm{MMSE}}\right]$; its solution
is the optimal encoder.

We also compare MSE-optimal encoding and decoding to three alternative
performance criteria.
\begin{itemize}
\item Maximization of Fisher information,
\item Minimization of the Bayesian Cramér-Rao bound,
\item Optimization of MSE (\ref{eq:mse}) for the Maximum Likelihood (ML)
estimator.
\end{itemize}
These performance criteria are defined and discussed in section \ref{subsec:Alternate-optimality-criteria}.

\subsection{Uniform-Gaussian coding\label{subsec:Uniform-Gaussian-coding}}

For mathematical tractability, we focus on the case of homogeneous
Gaussian tuning, which is a special case of that studied in \citep{Harel2018}:
the firing rate of the $i$th neuron in response to state $x$ is
given by 
\begin{equation}
\lambda^{i}\left(x\right)=h_{\mathrm{d}}\exp\left(-\frac{1}{2}\left\Vert x-\theta^{\left(i\right)}\right\Vert _{R}^{2}\right),\label{eq:gauss-tc-discrete}
\end{equation}
where $\theta^{\left(i\right)}\in\mathbb{R}^{m}$ is the $i$th neuron's
preferred stimulus, $h_{\mathrm{d}}\in\mathbb{R}_{+}$ is each neuron's
maximal expected firing rate, $R\in\mathbb{R}^{m\times m}$ is a symmetric
positive-definite matrix, and the notation $\left\Vert y\right\Vert _{M}^{2}$
denotes $y^{T}My$.  The distribution of the external state $X$
is assumed to be Gaussian, 
\begin{equation}
X\sim\mathcal{N}\left(\mu_{0},\Sigma_{0}\right).\label{eq:prior}
\end{equation}
\begin{figure}
\subfloat[uniform coding]{\includegraphics[width=0.5\columnwidth]{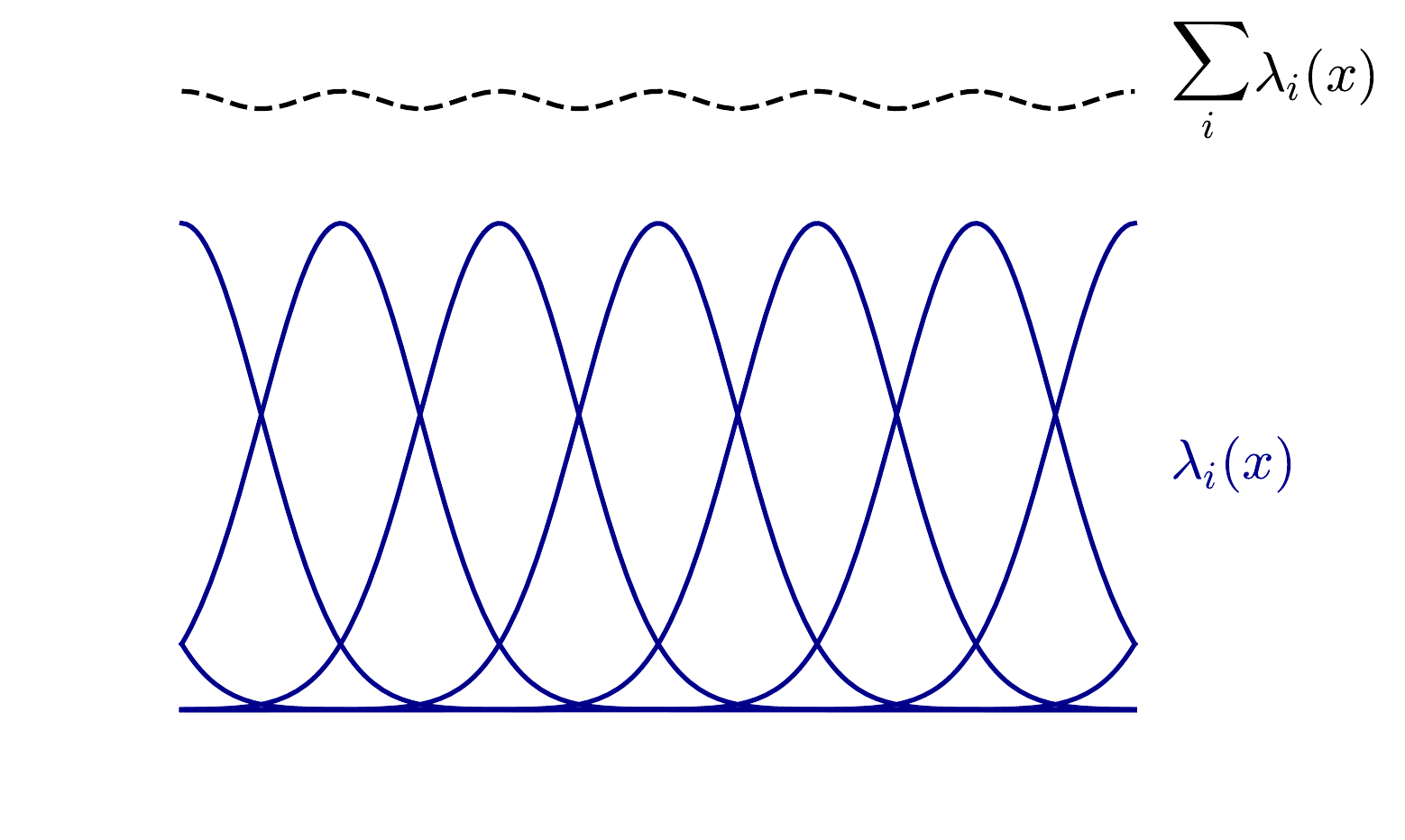}}\subfloat[non-uniform coding]{\includegraphics[width=0.5\columnwidth]{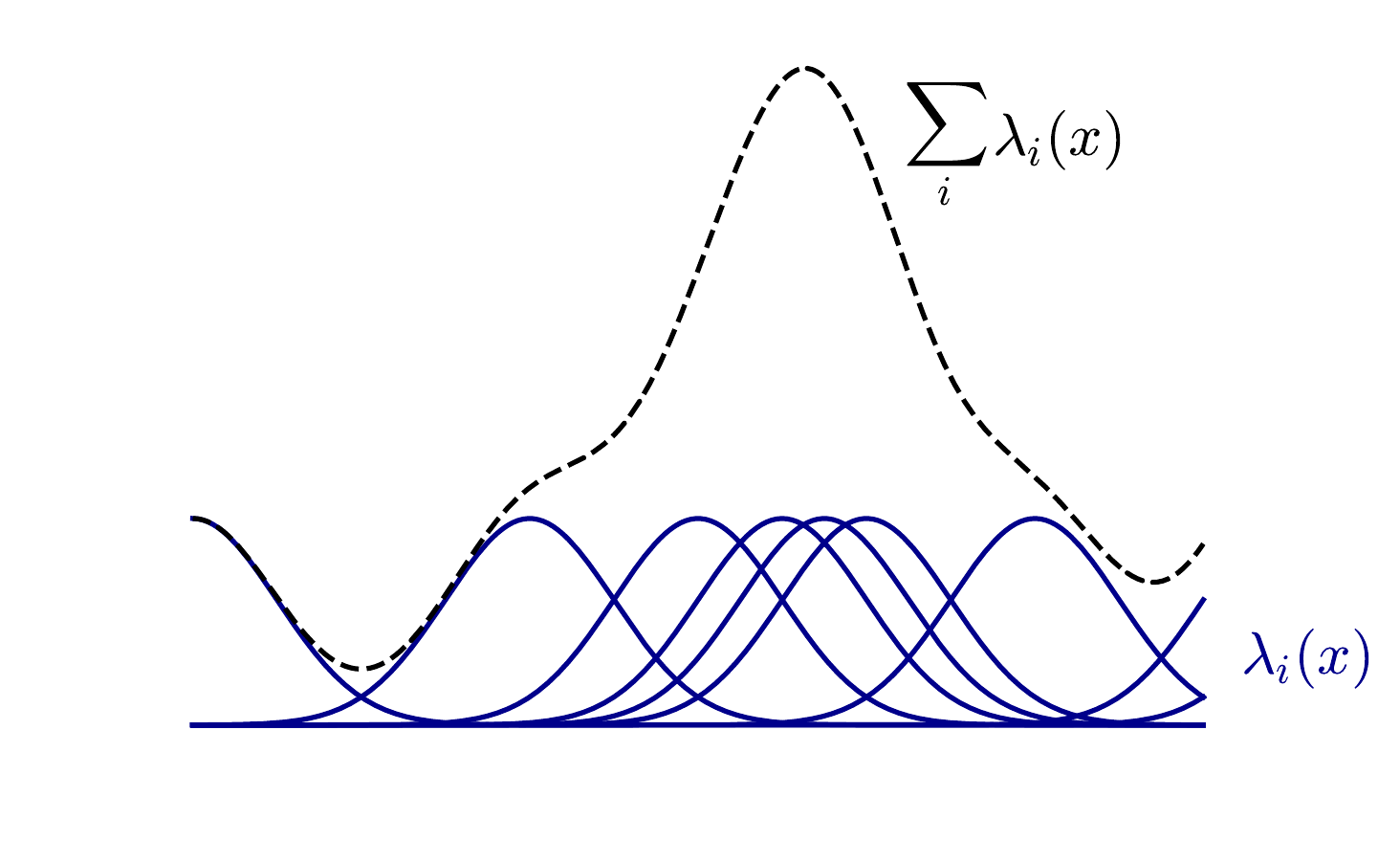}}\caption{The uniform coding property, $\Sigma_{i}\lambda^{i}\left(x\right)=\mathrm{const}$,
holds approximately in homogeneous populations with tuning function
centers located on a dense uniform grid, as demonstrated in (a) in
a 1-dimensional case. (b) illustrates non-uniform coding. Reprinted
from \citep{Harel2018}.}
\label{uniform}
\end{figure}

We assume preferred stimuli $\theta^{\left(i\right)}$ are spaced
uniformly on an infinite grid with spacing $\Delta\theta$, so there
is one neuron for each preferred stimulus of the form $\left(\alpha_{1},\ldots,\alpha_{m}\right)\transpose\Delta\theta$
where each $\alpha_{i}$ is an integer. When $\Delta\theta$ is reasonably
small relative to the eigenvalues of $R^{-1/2}$, the total firing
rate $\sum_{i}\lambda\left(x\right)$ is approximately independent
of $x$, a property that we refer to as \emph{uniform coding} (see
Figure \ref{uniform} and discussion in \citep{Harel2018}). For small
$\Delta\theta$, we may approximate the neural populations by an infinite
continuous population where the preferred stimulus $\theta$ may take
any value in $\mathbb{R}^{m}$ (see Figure \ref{continuous-pop}).
In this continuous population model, each spike is characterized by
the preferred stimulus of the firing neuron -- which is a continuous
variable -- rather than by the neuron's index. Thus, the firing pattern
is described by a \emph{marked point process} (MPP) $N$, which is
a random sequence of pairs $\left(t_{k},\theta_{k}\right)$, where
$t_{k}\in[0,\infty)$ is the time of the $k$th spike and $\theta_{k}\in\mathbb{R}^{m}$
its \emph{mark}, which is the preferred stimulus of the spiking neuron.
A marked point process is characterized by the rate of points with
marks in each (measurable) subset $\Theta\subseteq\mathbb{R}^{m}$,
which in the discrete population model (\ref{eq:gauss-tc-discrete})
is given by
\begin{align*}
\lambda\left(x;\Theta\right) & =\sum_{i:\theta_{i}\in\Theta}\lambda^{i}\left(x\right)=h_{\mathrm{d}}\sum_{i:\theta^{\left(i\right)}\in\Theta}\exp\left(-\frac{1}{2}\left\Vert x-\theta^{\left(i\right)}\right\Vert _{R}^{2}\right)\\
 & \approx\frac{h_{\mathrm{d}}}{\Delta\theta^{m}}\int_{\Theta}\exp\left(-\frac{1}{2}\left\Vert x-\theta\right\Vert _{R}^{2}\right)d\theta
\end{align*}
The approximate equality in the second line is exact in the limit
$\Delta\theta\to0$ with $h\coloneqq h_{\mathrm{d}}\Delta\theta^{-m}$
fixed. Accordingly, in the continuous population model we take the
MPP $N$ to have the space-time density 
\begin{equation}
\lambda\left(x;\theta\right)\coloneqq h\exp\left(-\frac{1}{2}\left\Vert x-\theta\right\Vert _{R}^{2}\right),\label{eq:gauss-tc}
\end{equation}
meaning that the rate of points with marks in $\Theta\subseteq\mathbb{R}^{m}$
is $\int_{\Theta}\lambda\left(x;\theta\right)d\theta$. In the continuous
population model, we use the notation $\boldsymbol{N}_{t}$ to refer
to the sequence of spike times and marks up to time $t$, that is
$\boldsymbol{N}_{t}\coloneqq\left(t_{k},\theta_{k}\right)_{k=1}^{N_{t}}$.
The probability density\footnote{$\boldsymbol{N}_{T}$ is a vector of random length $N_{T}$, and (\ref{eq:likelihood})
is its density for each value of $N_{T}$: when $N_{T}=n>0$, integrating
(\ref{eq:likelihood}) over the variables $\left\{ t_{k},\theta_{k}\right\} _{k=1}^{n}$
where $0\leq t_{1}\leq\cdots\leq t_{n}\leq T$ yields $\mathrm{P}\left(N_{T}=n\right)=e^{-rT}\left(rT\right)^{n}/n!$.
In the case $N_{T}=0$, the right-hand side of (\ref{eq:likelihood})
is itself the probability $\mathrm{P}\left(N_{T}=0\right)$, which
may be viewed as a density over the 0-dimensional space of possible
$\boldsymbol{N}_{T}$ values.} of $\boldsymbol{N}_{T}$ given $X=x$ is 
\begin{align}
p\left(\left(t_{k},\theta_{k}\right)_{k=1}^{N_{T}}|X=x\right) & =e^{-rT}\prod_{k=1}^{N_{T}}\lambda\left(x;\theta_{k}\right)\quad\left(t_{1}<\cdots<t_{N_{T}}\right)\label{eq:likelihood}
\end{align}
where 
\begin{equation}
r\coloneqq\int_{\mathbb{R}^{m}}\lambda\left(x;\theta\right)d\theta=h\sqrt{\frac{\left(2\pi\right)^{m}}{\det R}}\label{eq:total-rate}
\end{equation}
is the \emph{total population rate}. The likelihood (\ref{eq:likelihood})
does not depend on spike times $t_{k}$, only on preferred stimuli
$\theta_{k}$, so the sequence of preferred stimuli comprises a sufficient
statistic for the estimation of $X$.

\begin{figure}
\includegraphics[width=1\columnwidth]{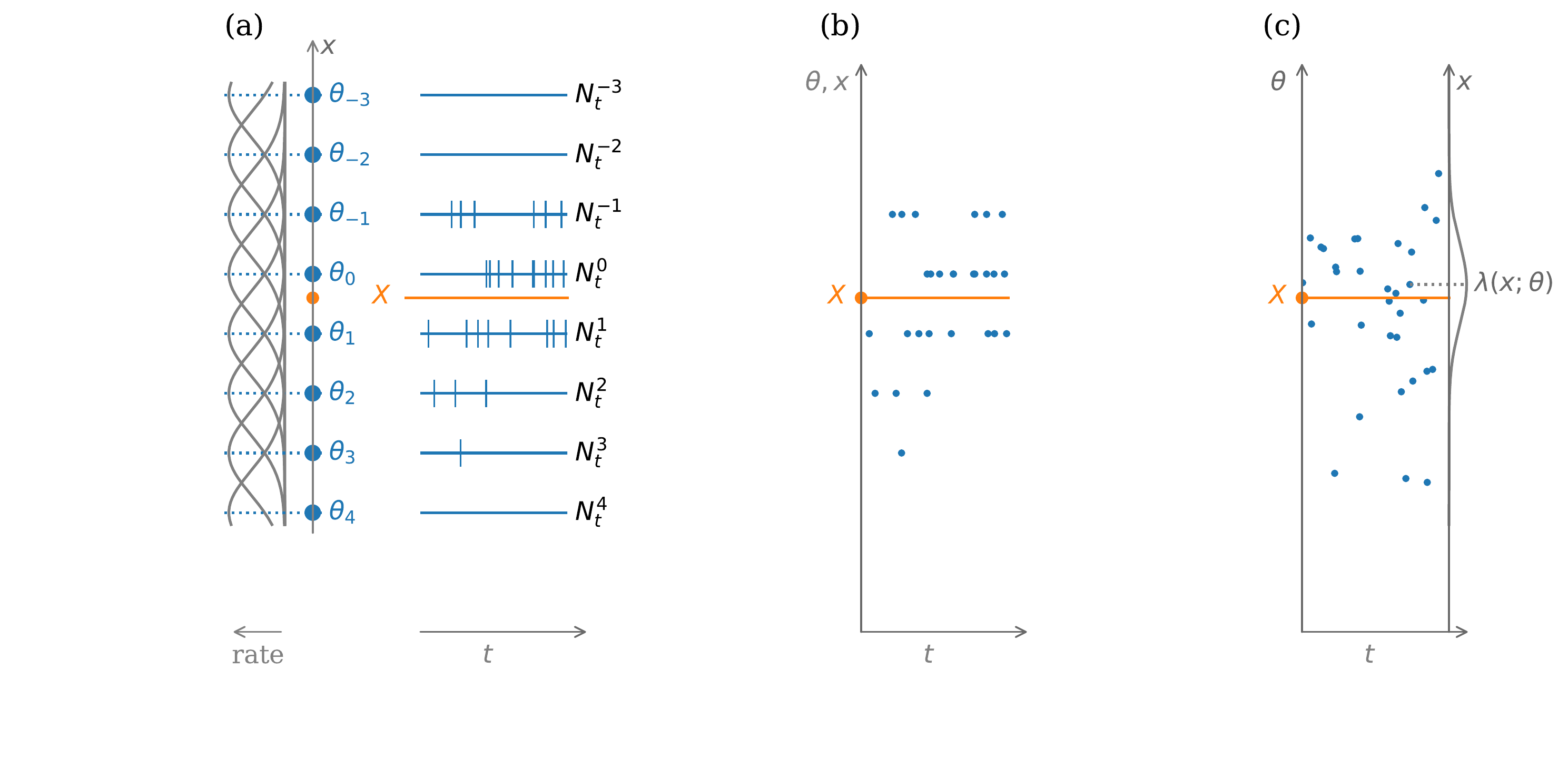}\caption{Discrete and continuous population models. \textbf{(a)} A discrete
population of neurons with preferred stimuli $\theta_{i}$ on a uniform
grid. The tuning functions are depicted on the left. The spiking pattern
of each neuron in response to the stimulus $X$ (orange line) is represented
by a separate point processes $N_{t}^{i}$ (blue vertical ticks).
\textbf{(b)} A representation of the same neural response as a single
marked point process. Each blue dot represents a spike with the vertical
location representing its mark $\theta$. \textbf{(c)} A continuous
population model approximating the statistics of the same population.
The approximation would be better for larger populations. The tuning
function corresponding to one of the spikes is depicted on the right.
Adapted from \citep{Harel2018}.}

\label{continuous-pop}
\end{figure}

\subsection{Energy constraints}

Clearly, the solution to the optimal encoding problem depends crucially
on the family of encodings over which optimization is performed. In
particular, if optimization is entirely unconstrained, the MMSE can
always be decreased by scaling the maximal rate-density $h$ in (\ref{eq:gauss-tc})
by some constant $k>1$, thereby uniformly increasing all firing rates.
The unconstrained problem is therefore mathematically ill-posed. However,
firing rates of real neurons cannot grow without bound. Although optimization
constraints are rarely discussed explicitly, all works mentioned above
-- apart from \citep{Bethge2002} -- implicitly constrain the neural
population in the same way: by fixing the maximal firing rate $h$
(or $h_{\mathrm{d}}$ in a discrete model) across all neurons and
optimizing only preferred stimuli and/or tuning widths. Note that
the continuous model approximates a discrete population with $h=h_{\mathrm{d}}/\Delta\theta^{m}$,
so that constraining the rate-density $h$ translates to a constraint
on maximal firing rates $h_{\mathrm{d}}$ of individual neurons.

We propose that in addition to a constraint on maximal rate-density
$h$, the \emph{expected} \emph{total }firing rate of the population
should also be constrained when formulating the optimal encoding problem.
Biologically, such a constraint corresponds to a constraint on the
expected rate of energy use by the neural population \citep{Zhang1999}.
As demonstrated below, this constraint may be crucial to make the
problem mathematically well-posed, even in the presence of a neuron-level
constraint on $h$. The energy-constrained problem takes the form
\begin{align}
\min_{h,R}\; & \epsilon_{\mathrm{MMSE}}=\epsilon\left(\mu_{T}\right)=\E\left[\mathrm{\mathrm{tr}}[(X-\mu_{T})(X-\mu_{T})\transpose]\right]\nonumber \\
\mathrm{s.t.}\; & h\leq\bar{h}\nonumber \\
 & r=h\sqrt{\frac{\left(2\pi\right)^{m}}{\det R}}\leq\bar{r}\label{eq:opt}
\end{align}
where $\mu_{T}\coloneqq\mathrm{E}\left[X|\boldsymbol{N}_{T}\right]$.

The population firing rate constraint may also be interpreted as a
sparsity constraint: for a fixed neuron-level rate constraint $\bar{h}$,
higher population rate $r$ is related to wider tuning resulting in
more neurons firing in response to the same stimulus. Low values of
the population constraint $\bar{r}$ thus enforce sparse coding, where
only a few neurons fire in response to each stimulus value.

\section{Closed-form computation of estimation error}

\subsection{General multivariate case}

The posterior distribution for the model (\ref{eq:prior})-(\ref{eq:likelihood})
is Gaussian, 
\[
X|\boldsymbol{N}_{T}\sim\mathcal{N}\left(\mu_{T},\Sigma_{T}\right)
\]
where

\begin{align}
\mu_{T} & =\Sigma_{T}\left(\Sigma_{0}^{-1}\mu_{0}+R\sum_{i=1}^{N_{T}}\theta_{i}\right),\nonumber \\
\Sigma_{T} & =\left(\Sigma_{0}^{-1}+N_{T}R\right)^{-1}.\label{eq:posterior}
\end{align}
This result is a special case of \citep[Theorem 1]{Snyder1977}. The
estimation MMSE is therefore
\begin{align*}
\epsilon_{\mathrm{MMSE}} & =\mathrm{\mathrm{tr}}\E\left[(X-\mu_{T})(X-\mu_{T})\transpose\right]=\mathrm{\mathrm{tr}}\E\E\left[(X-\mu_{T})(X-\mu_{T})\transpose|\boldsymbol{N}_{T}\right]\\
 & =\mathrm{tr}\,\E[\Sigma_{T}]=\mathrm{tr}\,\E\left[\left(\Sigma_{0}^{-1}+N_{T}R\right)^{-1}\right],
\end{align*}
where 
\[
N_{T}\sim\mathrm{Pois}\left(\sqrt{\frac{\left(2\pi\right)^{m}}{\det R}}hT\right)
\]
(see (\ref{eq:total-rate})), yielding the closed-form expression
\begin{equation}
\epsilon_{\mathrm{MMSE}}=\exp\left(-\sqrt{\frac{\left(2\pi\right)^{m}}{\det R}}hT\right)\sum_{k=0}^{\infty}\frac{1}{k!}\left(\sqrt{\frac{\left(2\pi\right)^{m}}{\det R}}hT\right)^{k}\mathrm{tr}\left((kR+\Sigma_{0}^{-1})^{-1}\right).\label{eq:mmse-uniform-gaussian}
\end{equation}

\subsection{Diagonal tuning}

An important special case of (\ref{eq:mmse-uniform-gaussian}) is
where the prior variance $\Sigma_{0}$ and tuning precision $R$ are
both diagonal, $m=n$, and $H=I$, with 
\begin{align}
\Sigma_{0} & =\diag\left(\sigma_{0,1}^{2},\sigma_{0,2}^{2},\ldots\sigma_{0,m}^{2}\right),\label{eq:diagonal-prior}\\
R & =\diag\left(\alpha_{1}^{-2},\ldots\alpha_{m}^{-2}\right).\label{eq:diagonal}
\end{align}
In section \ref{sec:Optimal-Encoding} we show that the general optimization
problem (\ref{eq:opt}), (\ref{eq:mmse-uniform-gaussian}) may be
reduced to this case, so there is no loss of generality here. In this
diagonal case, the total population rate is 
\begin{equation}
r=\left(2\pi\right)^{m/2}h\prod_{i}\alpha_{i},\label{eq:total-rate-diagonal}
\end{equation}
and (\ref{eq:mmse-uniform-gaussian}) reads
\begin{align}
\epsilon_{\mathrm{MMSE}} & =\exp\left(-h'T\prod_{j}\alpha_{j}\right)\sum_{k=0}^{\infty}\frac{(h'T\prod_{j}\alpha_{j})^{k}}{k!}\sum_{i=1}^{m}\frac{1}{\sigma_{0,i}^{-2}+k\alpha_{i}^{-2}},\nonumber \\
 & =\sum_{i=1}^{m}\sigma_{0,i}^{2}\,q\left(\frac{\alpha_{i}^{2}}{\sigma_{0,i}^{2}},h'T\prod_{j}\alpha_{j}\right),\label{eq:mmse-uniform-gaussian-diag}
\end{align}
where $h'\coloneqq\left(2\pi\right)^{m/2}h$ and we define, for $s,r>0$,
\begin{equation}
q\left(s,r\right)\coloneqq se^{-r}\sum_{k=0}^{\infty}\frac{r^{k}}{k!\left(s+k\right)}=M\left(1,s+1,-r\right),\label{eq:q}
\end{equation}
where $M$ is Kummer's confluent hypergeometric function \citep{andrews2000special}.
This result for the diagonal case has previously appeared in \citep{YaeMei10}.

In particular, the estimation error in the $i$th direction is given
by the $i$th summand in (\ref{eq:mmse-uniform-gaussian-diag}), namely
$\sigma_{0,i}^{2}\,q(\alpha_{i}^{2}/\sigma_{0,i}^{2},h'T\prod_{j}\alpha_{j})$.
The function $q$ takes values in the interval $\left[0,1\right]$,
so this factor represents the reduction of prior variance due to observation
of the spiking activity. More specifically, $q\left(s,r\right)$ satisfies
the bounds
\begin{equation}
\left(1+\frac{r}{s}\right)^{-1}\leq q\left(s,r\right)\leq\left(1+\frac{r}{s+1}\right)^{-1}.\label{eq:q-bounds}
\end{equation}
The lower bound is obtained by applying Jensen's inequality to the
convex function $k\mapsto s/\left(s+k\right)$ in (\ref{eq:q}). The
upper bound may be obtained from \citep[Theorem 15]{Luke1972inequalities}
by taking the limit $a\to1$ with $c=a+1$ in equation (5.4) therein.
This yields the following bounds for the MMSE

\begin{equation}
\sum_{i=1}^{m}\left(\frac{1}{\sigma_{0,i}^{2}}+\frac{h'T\prod_{j}\alpha_{j}}{\alpha_{i}^{2}}\right)^{-1}\leq\epsilon_{\mathrm{MMSE}}\leq\sum_{i=1}^{m}\left(\frac{1}{\sigma_{0,i}^{2}}+\frac{h'T\prod_{j}\alpha_{j}}{\alpha_{i}^{2}+\sigma_{0,i}^{2}}\right)^{-1}\label{eq:mmse-bound}
\end{equation}

We also use the following property of the function $q$, which is
proved by a direct calculation of the derivative, as outlined in appendix
\ref{sec:proofs}.

\begin{restatable}{clm}{qdecreasing}  

\label{claim:q-decreasing}For any $c_{1},c_{2}\in\mathbb{R}_{+}$
and $m\geq2$ , $q\left(c_{1}\alpha^{2},c_{2}\alpha^{m}\right)$ is
a decreasing function of $\alpha\in(0,\infty)$.

\end{restatable}

\section{Alternate performance criteria\label{subsec:Alternate-optimality-criteria}}

\subsection{Motivation}

Estimation MMSE is generally difficult to evaluate, except in specific
cases where the posterior distribution can be described analytically,
such as the model considered in the present work. This motivates the
optimization of more tractable approximations or bounds to the MMSE.
Of these, the most commonly used is Fisher information, which is related
to the MMSE through the Cramér-Rao bound (CRB) and the Bayesian Cramér-Rao
bound (BCRB) \citep{VanTrees2004Detection,YaeMei10}. Another class
of proxies is the MSE of suboptimal estimators such as the maximum
likelihood estimator (as used, e.g., in \citep{Finkelstein2018})
or the maximum a-posteriori estimator (as used, e.g., in \citep{WanStoLee2016}).
In section \ref{subsec:Two-dimensional-uniform-Gaussian} we compare
the minimization of MMSE in the model (\ref{eq:prior})-(\ref{eq:likelihood})
to optimization of several other performance criteria, in order to
evaluate their applicability as proxies to MMSE minimization.

\subsection{MSE of the maximum likelihood estimator\label{subsec:ML}}

Previous works such as \citep{Finkelstein2018} use the MSE of the
Maximum Likelihood (ML) estimator as a performance criterion. The
ML estimator may be easier to compute, especially for non-uniform
coding. However, it is not the optimal estimator, so optimization
of encoding for subsequent ML decoding does not optimize decoding
error of the encoder-decoder system. Using the ML estimator may be
justified for long decoding time or large firing rates, where the
effect of the prior on the posterior distribution becomes negligible.
However, in this regime, Fisher information may also be useful and
easier to evaluate. In section \ref{subsec:Two-dimensional-uniform-Gaussian},
we investigate the minimization of ML estimation error in our model
to assess whether it resembles minimization of MMSE for short decoding
times, at least qualitatively.

There is a difficulty with the definition of the ML estimator in the
uniform-Gaussian model: when there are no spikes, the likelihood (\ref{eq:likelihood})
is $\mathrm{P}\left(N_{T}=0|X=x\right)=e^{-rT}$, which is independent
of $x$, so the ML estimator is undefined. This is a direct consequence
of the uniform coding property, as the population rate $r$ is independent
of $x$. We therefore use a modified ML estimator, which equals the
prior $\mu_{0}$ in the case where there are no spikes,
\begin{equation}
\hat{X}_{t}^{\mathrm{ML}}=\begin{cases}
\frac{1}{N_{t}}\sum_{i=1}^{N_{t}}\theta_{i} & N_{t}>0,\\
\mu_{0} & N_{t}=0,
\end{cases}\label{eq:modified-ml}
\end{equation}
where $\theta_{i}$ is the preferred stimulus of the neuron firing
the $i$th spike. This modified ML estimator may be obtained from
the MMSE estimator in the limit $\Sigma_{0}^{-1}\to0$.

The MSE of this estimator can be written in closed form (see appendix
(\ref{subsec:app-ML})) as
\begin{align}
\epsilon\left(\hat{X}_{T}^{\mathrm{ML}}\right) & =e^{-rT}\left[\left(\sum_{k=1}^{\infty}\frac{\left(rT\right)^{k}}{k!\,k}\right)\sum_{i=1}^{m}\alpha_{i}^{2}+\sum_{i=1}^{m}\sigma_{i}^{2}\right],\label{eq:ml-mse}
\end{align}
where the second term is related to the case where there are no spikes.
For fixed $h$ and $r=\left(2\pi\right)^{m/2}h\prod_{i}\alpha_{i}$,
ML MSE (\ref{eq:ml-mse}) is minimized by minimizing $\sum_{i}\alpha_{i}^{2}$,
yielding the symmetric solution
\begin{equation}
\alpha_{1}=\alpha_{2}=\cdots=\alpha_{m}=\left(2\pi\right)^{-1/2}\left(\frac{r}{h}\right)^{1/m}.\label{eq:ml-mse-minimizer}
\end{equation}

\subsection{Fisher information / Cramér-Rao bound}

The CRB is a lower bound on the MSE of an \emph{unbiased} estimator
for a non-random parameter. In a Bayesian setting, it may be used
to bound the \emph{conditional }MSE of an unbiased estimator $\hat{X}$
conditioned on the random state $X$. Under suitable regularity conditions,
the bound reads \citep{VanTrees2004Detection}

\[
S\left(\hat{X}|X\right)\coloneqq\E\left[\left.\left(\hat{X}-X\right)\left(\hat{X}-X\right)\transpose\right|X\right]\succeq J\left(X\right)^{-1},
\]
where the relation $\succeq$ means that the difference $S\left(\hat{X}|X\right)-J\left(X\right)^{-1}$
is positive semidefinite, and $J\left(X\right)$ is the Fisher information
matrix, related to the likelihood function $p_{Y|X}$ as
\[
J_{ij}\left(x\right)\coloneqq-\E\left[\left.\frac{\partial^{2}\log p_{Y|X}\left(Y|x\right)}{\partial x_{i}\partial x_{j}}\right|X=x\right].
\]

As noted in the introduction, use of the Fisher information as a performance
criterion may be justified asymptotically, for large decoding times
or large population total firing rate. In this limit, the optimal
MSE approaches the CRB under mild conditions (see discussion in \citep{Bethge2002}).
Outside of the asymptotic regime, minimization of the CRB may differ
from MSE minimization since optimal estimators are typically biased,
and even among unbiased estimators, the bound may not be attainable.
Note that the continuous population model we study here is generally
\emph{not} within this asymptotic regime despite having infinitely
many neurons, since the total firing rate is kept finite.

In the case of a scalar state $x$ and a finite neural population
with tuning functions $\lambda_{k}$, Fisher information for the estimation
of $x$ from the spike counts over time interval $T$ is given by
\citep{DayAbb05} 
\[
J\left(x\right)=T\sum_{k}\frac{\left(\lambda'_{k}\left(x\right)\right)^{2}}{\lambda_{k}\left(x\right)}.
\]
Similarly, in the continuous population model with multivariate state,
(\ref{eq:gauss-tc})-(\ref{eq:likelihood}), the Fisher information
matrix takes the form
\begin{equation}
J\left(x\right)=T\int\frac{\nabla_{x}\lambda\left(x;\theta\right)\nabla_{x}\lambda\left(x;\theta\right)\transpose}{\lambda\left(x;\theta\right)}d\theta=h'T\left|R\right|^{-1/2}R=rTR,\label{eq:uniform-gaussian-fisher}
\end{equation}
(see derivation in appendix, section \ref{sec:Fisher-infinite}) and
the MSE of any \emph{unbiased }estimator is therefore bounded by 
\begin{equation}
\epsilon\left(\hat{X}\right)=\E\,\mathrm{tr}\,S\left(\hat{X}|X\right)\geq\E\,\mathrm{tr}\left(J^{-1}\right)=\frac{1}{rT}\mathrm{tr}\left(R^{-1}\right)=\frac{\sum_{i}\alpha_{i}^{2}}{rT}=\frac{\sum_{i}\alpha_{i}^{2}}{h'T\prod_{i}\alpha_{i}},\label{eq:inverse-fisher}
\end{equation}
For fixed population rate $r$, this bound is minimized by the same
tuning widths (\ref{eq:ml-mse-minimizer}) as the MSE of the modified
ML estimator (\ref{eq:modified-ml}). This results in the minimal
population inverse Fisher information

\begin{equation}
\min\E\,\mathrm{tr}\left(J^{-1}\right)=\frac{m}{h'T}\alpha^{2-m},\label{eq:inverse-fisher-opt}
\end{equation}
where the tuning width $\alpha$ is related to rate density $h$ and
population rate $r$ through (\ref{eq:ml-mse-minimizer}). Equation
(\ref{eq:inverse-fisher-opt}) demonstrates that in the absence of
population rate constraint, the Fisher-optimal widths are zero in
the univariate case $m=1$ and infinite for $m>2$, whereas in the
two-dimensional case $m=2$, population Fisher information is independent
of tuning width. This is a special case of a result from \citep{Zhang1999}
which applies more generally to radially symmetric tuning. This preference
for wide tuning in higher dimensions may be explained by the fact
that encoding accuracy in each dimension benefits from wide tuning
in other dimensions, due to an increase in population firing rate,
as noted in \citep{Eurich2000} (see Figure 1 therein).

Note that the modified ML estimator (\ref{eq:modified-ml}) is not
an unbiased estimator; its bias is $\left(\mu_{0}-X\right)P\left\{ N_{t}=0|X\right\} =\left(\mu_{0}-X\right)e^{-rT}$.
Therefore its MSE is not bounded by (\ref{eq:inverse-fisher}), though
we may expect it to satisfy the bound when $P\left(N_{t}=0\right)$
is small.

\subsection{Bayesian Cramér-Rao bound}

The Bayesian Cramér-Rao Bound (BCRB) is a lower bound on estimation
error in a Bayesian setting. Under suitable regularity conditions,
the MSE of any estimator $\hat{X}$ is bounded as \citep{VanTrees2004Detection}
\begin{equation}
S\left(\hat{X}\right)\coloneqq\E\left[\left(\hat{X}-X\right)\left(\hat{X}-X\right)\transpose\right]\succeq\left(\E\left[J\left(X\right)\right]+J_{\mathrm{P}}\right)^{-1},\label{eq:bcrb}
\end{equation}
where $J\left(X\right)$ is the Fisher information matrix and $J_{\mathrm{P}}$
is the Bayesian information matrix, related to the prior density $p_{X}$
by
\[
\left[J_{\mathrm{P}}\right]_{ij}=-\E\left[\frac{\partial^{2}\log p_{X}\left(X\right)}{\partial x_{i}\partial x_{j}}\right].
\]
The BCRB is not restricted to unbiased estimators and is therefore
applicable as a bound on optimal estimation error in a Bayesian setting
outside of the asymptotic regime of large decoding time or population
firing rate. On the other hand, unlike the CRB, the BCRB is generally
not attained asymptotically: in the limit of infinite decoding time,
the MMSE approaches the expected value of the CRB, $\E[J\left(X\right)^{-1}]$,
whereas the BCRB approaches $\E\left[J\left(X\right)\right]^{-1}$,
which generally underestimates the asymptotic MMSE \citep{VanTrees2004Detection}.
In other words, the ratio $\mathrm{MMSE}/\mathrm{BCRB}$ may approach
a value strictly larger 1 in the infinite decoding time limit. As
evident from (\ref{eq:bcrb}), in the univariate case minimization
of the BCRB is equivalent to maximization of Fisher information. This
equivalence does not hold in higher dimensions, as demonstrated below
in the 2-dimensional case.

In the continuous population model (\ref{eq:prior})-(\ref{eq:likelihood}),
the BCRB reads

\[
S\left(\hat{X}\right)\succeq\left(h'T\left|R\right|^{-1/2}R+\Sigma_{0}^{-1}\right)^{-1},
\]
where we have used (\ref{eq:uniform-gaussian-fisher}). Applying the
bound to the MMSE estimator in the diagonal case, 
\begin{align}
\epsilon_{\mathrm{MMSE}}=\mathrm{tr}\,S\left(\hat{X}_{\mathrm{MMSE}}\right) & \geq\sum_{i=1}^{m}\left(\sigma_{0,i}^{-2}+h'T\left(\prod_{j}\alpha_{j}\right)\alpha_{i}^{-2}\right)^{-1}\label{eq:uniform-gaussian-bcrb}\\
 & =\sum_{i=1}^{m}\left(\sigma_{0,i}^{-2}+rT\alpha_{i}^{-2}\right)^{-1},
\end{align}
which is the same lower bound as in (\ref{eq:mmse-bound}).

\section{Optimal encoding\label{sec:Optimal-Encoding}}

\subsection{MMSE-optimal multivariate tuning\label{subsec:MMSE-optimal-multivariate-tuning}}

Substituting (\ref{eq:mmse-uniform-gaussian}) into (\ref{eq:opt}),
we obtain the optimization problem
\begin{align}
\min_{R,h'}\quad & \exp\left(-\frac{h'T}{\sqrt{\det R}}\right)\sum_{k=0}^{\infty}\frac{1}{k!}\left(\frac{h'T}{\sqrt{\det R}}\right)^{k}\mathrm{tr}\left((kR+\Sigma_{0}^{-1})^{-1}\right),\label{eq:opt-multi}\\
\mathrm{s.t.}\quad & h'\leq\bar{h}',\nonumber \\
 & \frac{h'}{\sqrt{\det R}}\leq\bar{r},\nonumber 
\end{align}
Denote the spectral decomposition for the prior variance by $\Sigma_{0}\coloneqq V_{0}\Lambda_{0}V_{0}\transpose$
with $\Lambda_{0}$ diagonal and $V_{0}$ orthogonal. The principal
components of the random vector $X$ are $\tilde{X}\coloneqq V_{0}\transpose X$,
which has the diagonal prior variance $\E[\tilde{X}\tilde{X}\transpose]=\Lambda_{0}$.
The population rate-density (\ref{eq:gauss-tc}) may be written in
terms of the principal components as
\begin{align}
\tilde{\lambda}\left(\tilde{x};\tilde{\theta}\right) & \coloneqq h\exp\left(-\frac{1}{2}\left\Vert \tilde{x}-\tilde{\theta}\right\Vert _{\tilde{R}}^{2}\right)=\lambda\left(x;\theta\right),\label{eq:gauss-tc-pc}
\end{align}
where 
\begin{align*}
\tilde{R} & =V_{0}\transpose RV_{0},\\
\tilde{x} & =V_{0}\transpose x,\quad\tilde{\theta}=V_{0}\transpose\theta.
\end{align*}
The optimization problem may be reformulated as minimization of decoding
MMSE for the principal component vector $\tilde{X}$, since neither
the objective function nor the constraints in (\ref{eq:opt}) are
affected by the orthogonal transformation $V_{0}$. Specifically,
the decoding MMSE of $\tilde{X}$ is $\mathrm{tr}[V_{0}E_{X}V_{0}\transpose]=\mathrm{tr}\left[V\transpose_{0}V_{0}E_{X}\right]=E_{X}$
where $E_{X}\coloneqq\E[(X-\hat{X})(X-\hat{X})]\transpose$ is the
decoding MMSE of $X$; and the optimization constraints depend on
$R$ only through $\det R=\det\tilde{R}$. We therefore \emph{assume
from here on, without loss of generality, that $\Sigma_{0}$ is diagonal}
with decreasing elements,
\begin{align*}
\Sigma_{0} & =\diag\left(\sigma_{0,1}^{2},\sigma_{0,2}^{2},\ldots\sigma_{0,m}^{2}\right),\\
 & \sigma_{0,1}^{2}\geq\sigma_{0,2}^{2}\geq\cdots\geq\sigma_{0,m}^{2}.
\end{align*}

The problem is further reduced by writing the spectral decomposition
$R=U\Lambda U\transpose$ and $Q\coloneqq\Sigma_{0}^{-1}$, and using
the following claim, proved in appendix \ref{sec:proofs}
\begin{restatable}{clm}{trInv}\label{claim:tr-inv}

Let $\Lambda,Q$ be diagonal positive definite matrices, $\Lambda=\mathrm{diag}\left(\lambda_{i}\right),Q=\mathrm{diag}\left(q_{i}\right)$
where $0<\lambda_{1}\leq\lambda_{2}\leq\cdots\lambda_{n}$ and $0<q_{1}\leq q_{2}\leq\cdots\leq q_{n}$.
Then the constrained minimization problem
\begin{align*}
\min_{U}\; & \mathrm{tr}\left[\left(Q+U\Lambda U\transpose\right)^{-1}\right]\\
\mathrm{s.t.}\; & UU\transpose=I
\end{align*}
is solved by $U$ anti-diagonal $u_{ij}=\delta_{i,\left(n+1-j\right)}$.

\end{restatable}

This shows that for a fixed choice of eigenvalues for $R$, each term
in (\ref{eq:opt-multi}) is minimized by the same choice of eigenvectors,
namely antidiagonal $u_{ij}=\delta_{i,\left(n+1-j\right)}$, which
makes $R$ a diagonal matrix with elements sorted in decreasing order.
Therefore this $U$ also minimizes the sum in (\ref{eq:opt-multi})
for fixed eigenvalues. Since the constraints depend on $R$ only through
its eigenvalues, the optimal solution is of the same form, and in
particular $R$ is diagonal and tuning is narrower (larger eigenvalues
of $R$) in directions with larger prior variance (smaller eigenvalues
of $Q$). In the more general case where $X$ has correlated components,
applying this result to the principal components $\tilde{X}$ as outlined
above shows that optimal tuning is aligned to principal components,
$R=V_{0}\tilde{R}V_{0}\transpose$ where $\tilde{R}$ is diagonal
and principal components are in the columns of $V_{0}$.

We have reduced the problem to the case of diagonal prior $\Sigma_{0}=\diag\left(\sigma_{0,1}^{2},\sigma_{0,2}^{2},\ldots\sigma_{0,m}^{2}\right)$
and tuning $R=\diag\left(\alpha_{1}^{-2},\ldots\alpha_{m}^{-2}\right)$.
Substituting (\ref{eq:mmse-uniform-gaussian-diag}) into (\ref{eq:opt}),
we obtain the optimization problem
\begin{align*}
\min_{\boldsymbol{\alpha},h'}\quad & \sum_{i=1}^{m}\sigma_{0,i}^{2}q\left(\frac{\alpha_{i}^{2}}{\sigma_{0,i}^{2}},h'T\prod_{j}\alpha_{j}\right),\\
\mathrm{s.t.}\quad & h'\leq\bar{h}',\\
 & h'\prod_{i}\alpha_{i}\leq\bar{r},
\end{align*}
where $\boldsymbol{\alpha}=\left(\alpha_{1},\ldots\alpha_{m}\right),h'=\left(2\pi\right)^{m/2}h$,
$\bar{h}'=\left(2\pi\right)^{m/2}\bar{h}$, and $q$ is given by (\ref{eq:q}).
The problem may be re-parameterized in terms of $\boldsymbol{\alpha}$
and $r$, where $r=h'\prod\alpha_{i}$ is the total population rate,
\begin{align}
\min_{\boldsymbol{\alpha},r}\quad & \sum_{i=1}^{m}\sigma_{0,i}^{2}q\left(\frac{\alpha_{i}^{2}}{\sigma_{0,i}^{2}},rT\right),\label{eq:opt-r-alpha}\\
\mathrm{s.t.}\quad & r\leq\bar{h}'\prod_{i}\alpha_{i},\nonumber \\
 & r\leq\bar{r}.\nonumber 
\end{align}
Note that $q\left(s,r\right)$ is increasing in $s$ for fixed $r$,
as evident from (\ref{eq:q}). Therefore, for fixed $r$, the objective
(\ref{eq:opt-r-alpha}) is an increasing function of each $\alpha_{i}$,
so the solution satisfies the first constraint with equality, $\prod_{i}\alpha_{i}=r/\bar{h}'$,
or equivalently $h'=\bar{h}'$, and the problem reduces to 
\begin{align}
\min_{\boldsymbol{\alpha}\in\mathbb{R}^{m}}\quad & \sum_{i=1}^{m}\sigma_{0,i}^{2}q\left(\frac{\alpha_{i}^{2}}{\sigma_{0,i}^{2}},\bar{h}'T\prod_{j}\alpha_{j}\right)\label{eq:reduced-obj}\\
\mathrm{s.t.}\quad & \prod_{i}\alpha_{i}\leq\frac{\bar{r}}{\bar{h}'}.\label{eq:reduced-constraint}
\end{align}
The unconstrained version of this problem in the scalar case $m=1$
has been studied in \citep{YaeMei10}, where it was shown that optimal
tuning width decreases with encoding time and increases with prior
variance, consistent with experimental results. In this scalar case,
the addition of the constraint (\ref{eq:reduced-constraint}) modifies
the solution in a simple way by fixing the tuning width $\alpha$
to its constraint $\bar{r}/\bar{h}'$ when the unconstrained problem
is solved by a larger value of $\alpha$.

Using claim \ref{claim:q-decreasing}, for $m\geq2$, any scaling
of all $\alpha_{i}$ by the same factor $>1$ would reduce the MMSE,
so the remaining constraint is also satisfied with equality, yielding
the problem
\begin{align}
\min_{\boldsymbol{\alpha}\in\mathbb{R}^{m}}\quad & \sum_{i=1}^{m}\sigma_{0,i}^{2}q\left(\frac{\alpha_{i}^{2}}{\sigma_{0,i}^{2}},\bar{r}T\right)\quad\left(m\geq2\right)\label{eq:diagonal-obj}\\
\mathrm{s.t.}\quad & \prod_{i}\alpha_{i}=\frac{\bar{r}}{\bar{h}'}\label{eq:diagonal-constraint}
\end{align}
Similarly, claim \ref{claim:q-decreasing} implies that the unconstrained
version of (\ref{eq:reduced-obj}) has no finite solution for $m\geq2$,
as the minimal error is attained in the limit of infinitely wide tuning.
Equivalently, the minimization of the MMSE with neuron-level constraint
only and no population-level constraint is solved in the limit of
infinitely wide tuning, demonstrating the importance of the population-level
constraint in this model.

To understand this difference between the scalar case $m=1$ and vector
case $m\geq2$, note that broadening the tuning has two effects: \emph{(i)}
increasing the spike rate, and \emph{(ii)} making each spike less
informative. Specifically, the population spike rate is proportional
to the product of all tuning widths, $r=h'\prod_{i}\alpha_{i}$, so
that wider tuning yields higher population rates. On the other hand,
(\ref{eq:posterior}) in the diagonal case reads 
\[
\sigma_{t,i}^{-2}=\sigma_{0,i}^{-2}+N_{t}\alpha_{i}^{-2},
\]
where $\sigma_{t,i}^{-2}$ is the posterior precision (inverse of
variance) of $X_{i}$ conditioned on $\boldsymbol{N}_{t}$, so that
narrower tuning yields more informative spikes. In the scalar case,
$m=1$, the trade-off between these two effects leads to a finite
optimal tuning width. In higher dimensions, $m\geq2$, all tuning
widths contribute to the firing rate, but the effect of each spike
on the posterior precision in dimension $i$ depends only on $\alpha_{i}$.
Thus, optimizing the tuning width $\alpha_{i}$ involves a trade-off
for the posterior variance $\sigma_{t,i}^{2}$ of $X_{i}$, but affects
$\sigma_{t,j}^{2}$ for any other direction $j\neq i$ monotonically:
increasing the tuning width $\alpha_{i}$ can only decrease $\sigma_{t,j}^{2}$
for $j\neq i$, through its effect on the population rate. This shifts
the trade-off towards wider tuning in higher dimensions.

This effect can be seen more explicitly by approximating the objective
(\ref{eq:reduced-obj}) using the bounds (\ref{eq:q-bounds}) as
\[
\epsilon'\left(\boldsymbol{\alpha};0\right)\leq\mathrm{MMSE}\leq\epsilon'\left(\boldsymbol{\alpha};1\right)
\]
where
\[
\epsilon'\left(\boldsymbol{\alpha};\beta\right)\coloneqq\sum_{i=1}^{m}\left(\sigma_{0,i}^{-2}+\frac{\bar{h}'T\prod_{j}\alpha_{j}}{\alpha_{i}^{2}+\beta\sigma_{0,i}^{2}}\right)^{-1}.
\]
Consider the effect of scaling the tuning width by a common factor
$c$ in each dimension,
\begin{equation}
\epsilon'\left(c\alpha_{1},\ldots,c\alpha_{m};\beta\right)=\sum_{i=1}^{m}\left(\sigma_{0,i}^{-2}+\bar{r}T\frac{c^{m}}{c^{2}\alpha_{i}^{2}+\beta\sigma_{0,i}^{2}}\right)^{-1}.\label{eq:bound-scaling}
\end{equation}
As a function of the scaling $c$, the lower bound ($\beta=0$) is
increasing when $m=1$, independent of $c$ when $m=2$ and decreasing
when $m>2$. The upper bound ($\beta=1$) is decreasing for $m\geq2$;
for $m=1$ each term in the upper bound ($\beta=1$) is minimized
at the finite value $c=\sigma_{0,i}^{2}/\alpha_{i}^{2}$ and is maximized
at $c=0$ and $c\to\infty$.

\subsection{Two-dimensional uniform-Gaussian encoding\label{subsec:Two-dimensional-uniform-Gaussian}}

Our analysis indicates that there is a qualitative difference between
univariate coding, where MMSE is minimized at a finite tuning width
for fixed maximal rate-density $h$, and multivariate coding, where
an additional population-level constraint is necessary to achieve
finite width. To study the effect of this constraint, we numerically
analyze the solution to the constrained optimization problem (\ref{eq:diagonal-obj})-(\ref{eq:diagonal-constraint})
in the simplest multivariate setting of two dimensions. In this setting,
the parameters optimized are two tuning widths $\alpha_{1},\alpha_{2}$,
with the constraint (\ref{eq:diagonal-constraint}) removing one degree
of freedom, reducing the problem to a univariate optimization problem.
Defining 
\[
\gamma_{i}\coloneqq\frac{\alpha_{i}}{\alpha_{1}+\alpha_{2}},
\]
the problem may be formulated as choosing the optimal $\gamma_{1}\in\left[0,1\right]$.
Values of $\gamma_{1}$ near 0 correspond to narrow tuning in dimension
1 and wide tuning in dimension 2, which may be interpreted as ``one-dimensional''
encoding focusing on dimension 1. Similarly, values of $\gamma_{1}$
near 1 correspond to ``one-dimensional encoding'' of dimension 2.
The case $\gamma_{1}=\frac{1}{2}$ corresponds to the same tuning
width in both directions (``two-dimensional encoding'').

Figure \ref{fig:encoding-time} illustrates MMSE as well as alternate
performance criteria described in section \ref{subsec:Alternate-optimality-criteria},
for two-dimensional encoding in the diagonal case (\ref{eq:diagonal-obj}),
with fixed neuron-level and population-level energy constraints. The
prior variance $\Sigma_{0}$ is asymmetric: $\sigma_{0,1}^{2}<\sigma_{0,2}^{2}$.
Figure \ref{fig:criteria-time} shows the value of the performance
criteria as a function of encoding time $T$ and the tuning width
ratio $\gamma_{1}=\alpha_{1}/\left(\alpha_{1}+\alpha_{2}\right)$.
The optimal $\gamma_{1}$ is marked with a gray line, and optimal
encoding is illustrated in Figure \ref{fig:opt-grid-time} for several
encoding times. Note that the energy constraints and tuning width
ratio $\gamma_{1}$ together determine the widths $\alpha_{1},\alpha_{2}$
through the constraint (\ref{eq:diagonal-constraint}) on the product
$\alpha_{1}\alpha_{2}$. MSE-optimal tuning solving (\ref{eq:diagonal-obj})-(\ref{eq:diagonal-constraint})
is narrower in the dimension that has greater prior variance, $\alpha_{1}^{2}>\alpha_{2}^{2}$,
as shown analytically in Claim \ref{claim:tr-inv}. The size of this
effect decreases with encoding time. Optimization of the Bayesian
Cramér-Rao bound (\ref{eq:uniform-gaussian-bcrb}) shows similar behavior.
Optimization of both the ML estimator's MSE and of the (non-Bayesian)
Cramér-Rao bound yields symmetric tuning $\gamma_{1}=\frac{1}{2}$
regardless of decoding time (see (\ref{eq:ml-mse-minimizer})) --
a qualitatively incorrect result.

\begin{figure}
\centering{}\captionsetup[subfigure]{position=top}\subfloat[]{\includegraphics[height=0.3\columnwidth]{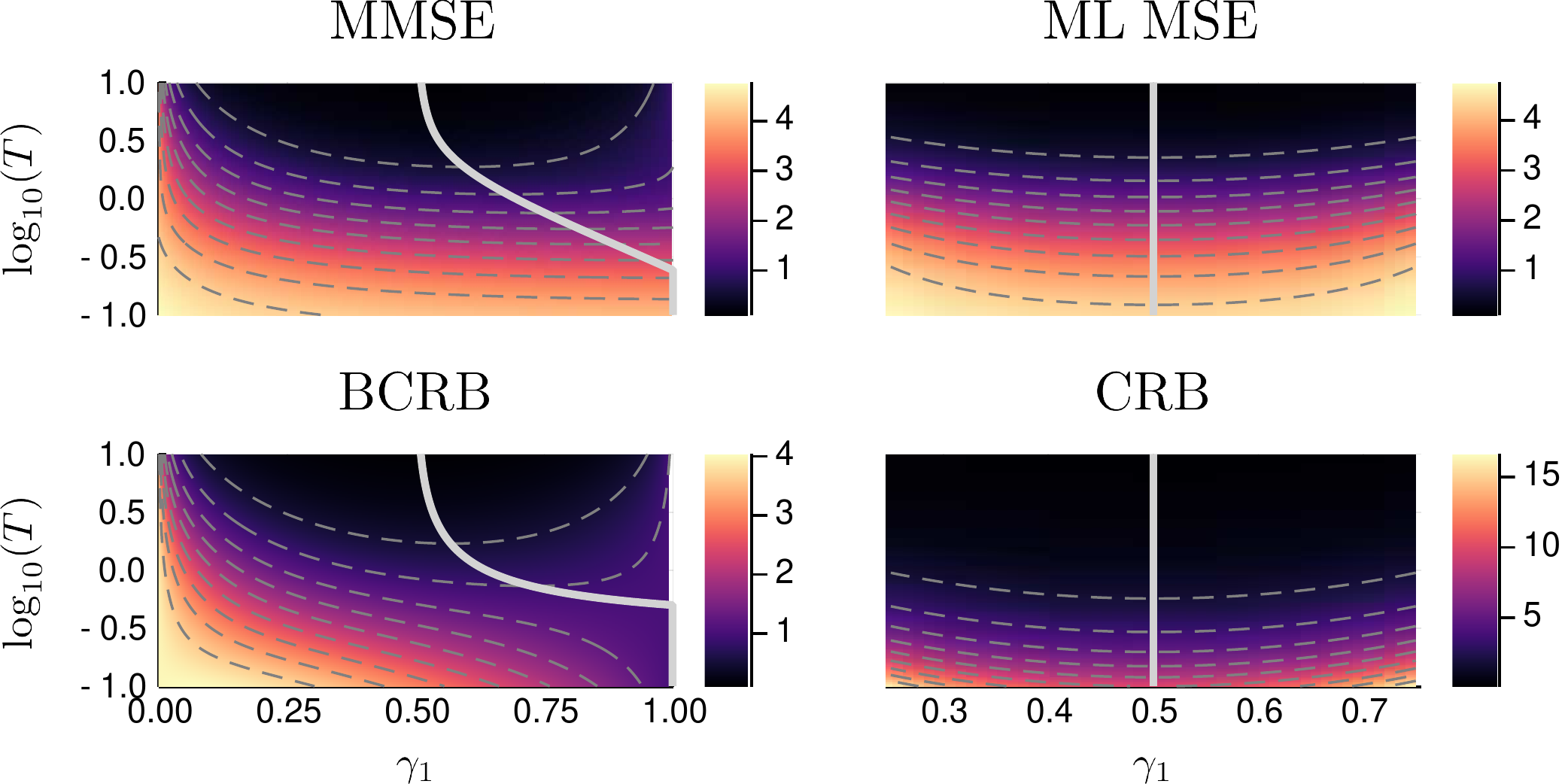}\label{fig:criteria-time}}\hfill{}\subfloat[]{\includegraphics[height=0.3\columnwidth]{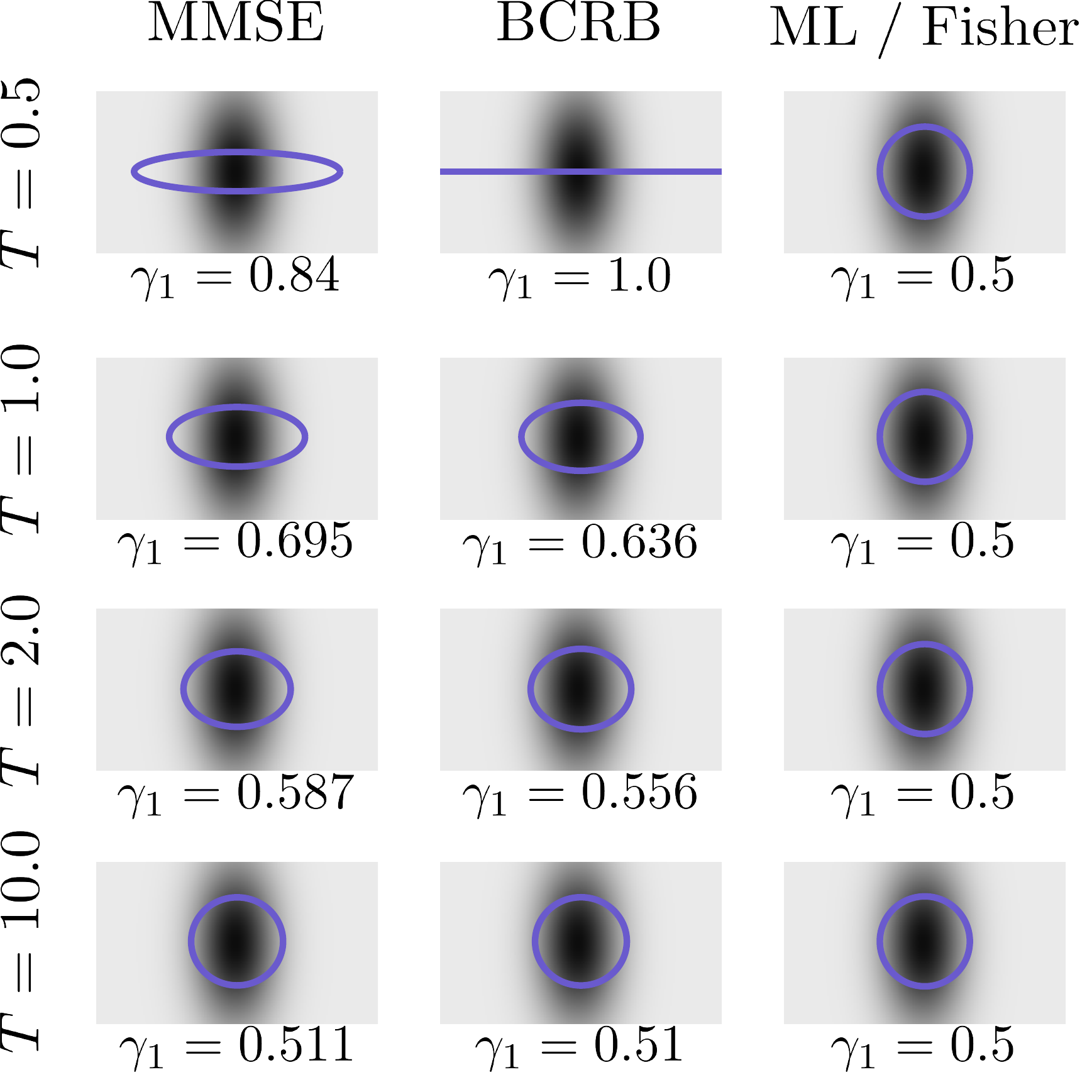}\label{fig:opt-grid-time}}\caption{Encoding performance and optimal encoding in a two-dimensional uniform-Gaussian
population with asymmetric prior $\sigma_{0,1}^{2}<\sigma_{0,2}^{2}$.
Prior distribution as well as neuron-level and population-level energy
constraints are the same in all subplots.\textbf{ (a)} Performance
criteria as a function of decoding time $T$ and tuning width ratio
$\gamma_{1}=\alpha_{1}/\left(\alpha_{1}+\alpha_{2}\right)$. Contour
levels are shown as dashed lines. Solid gray line indicates optimal
width ratio $\gamma_{1}$. For the ML MSE and CRB, the $\gamma_{1}$
axis is clipped since these criteria approach infinity in each of
the limits $\gamma_{1}\to0$, $\gamma_{1}\to1$. \textbf{(b) }Optimal
encoding according to the various criteria (columns) over several
decoding times (rows). A level set of the optimal tuning function
is shown in blue, superimposed over the prior density. Optimal tuning
width ratio is noted below each subplot. Parameters in all plots:
$\bar{r}=2.5,2\pi\bar{h}=2,\Sigma_{0}=\mathrm{diag}\left(1,4\right)$.}
\label{fig:encoding-time}
\end{figure}

The effect of the population-level energy constraint $\bar{r}$ is
explored in Figure \ref{fig:encoding-rate}, where the optimality
criteria are plotted as a function of $\bar{r}$ and the tuning width
ratio $\gamma_{1}$, with the neuron-level constraint $\bar{h}$ and
encoding time $T$ fixed. The prior variance $\Sigma_{0}$ in Figure
\ref{fig:encoding-rate} is narrower in dimension 1: $\sigma_{0,1}^{2}<\sigma_{0,2}^{2}$.
As in Figure \ref{fig:criteria-time}, this produces an asymmetry
in the opposite direction in the optimal tuning widths $\alpha_{1}^{2}>\alpha_{2}^{2}$,
or equivalently, $\gamma_{1}>\frac{1}{2}$. Tuning is almost two-dimensional
for low population rate constraint $\bar{r}$, becomes one-dimensional
at intermediate population rates, and nearly two-dimensional again
for high population rates. The BCRB correctly predicts $\gamma_{1}>\frac{1}{2}$
but fails to capture the dependence on the population constraint $\bar{r}$.
The fact that the BCRB does not depend on $\bar{r}$ in the two-dimensional
case may be seen by rewriting (\ref{eq:uniform-gaussian-bcrb}) for
the case $m=2$ as

\begin{align*}
\mathrm{BCRB} & =\left(\sigma_{0,1}^{-2}+h'T\frac{\alpha_{2}}{\alpha_{1}}\right)^{-1}+\left(\sigma_{0,2}^{-2}+h'T\frac{\alpha_{1}}{\alpha_{2}}\right)^{-1}\\
 & =\left(\sigma_{0,1}^{-2}+h'T\left(\gamma_{1}^{-1}-1\right)\right)^{-1}+\left(\sigma_{0,2}^{-2}+\frac{h'T}{\left(\gamma_{1}^{-1}-1\right)}\right)^{-1}
\end{align*}
which for fixed $h',T$ is strictly a function of $\gamma_{1}$.

Figure \ref{fig:encoding-opt} shows the dependence of MSE-optimal
tuning width ratio on prior asymmetry, quantified as the ratio $\sigma_{0,1}/\left(\sigma_{0,1}+\sigma_{0,2}\right)$,
and on encoding time (left) or population-level energy constraint
(right). Optimal tuning always has the opposite asymmetry to the prior.
As seen on the left plot, optimal tuning is one-dimensional for short
decoding time and becomes increasingly symmetric with increasing decoding
time.

\begin{figure}
\captionsetup[subfigure]{position=top}

\subfloat[]{\includegraphics[height=0.3\columnwidth]{MMSE_ML_BCRB_Fisher_vs_time}

}\hfill{}\subfloat[]{\includegraphics[height=0.3\columnwidth]{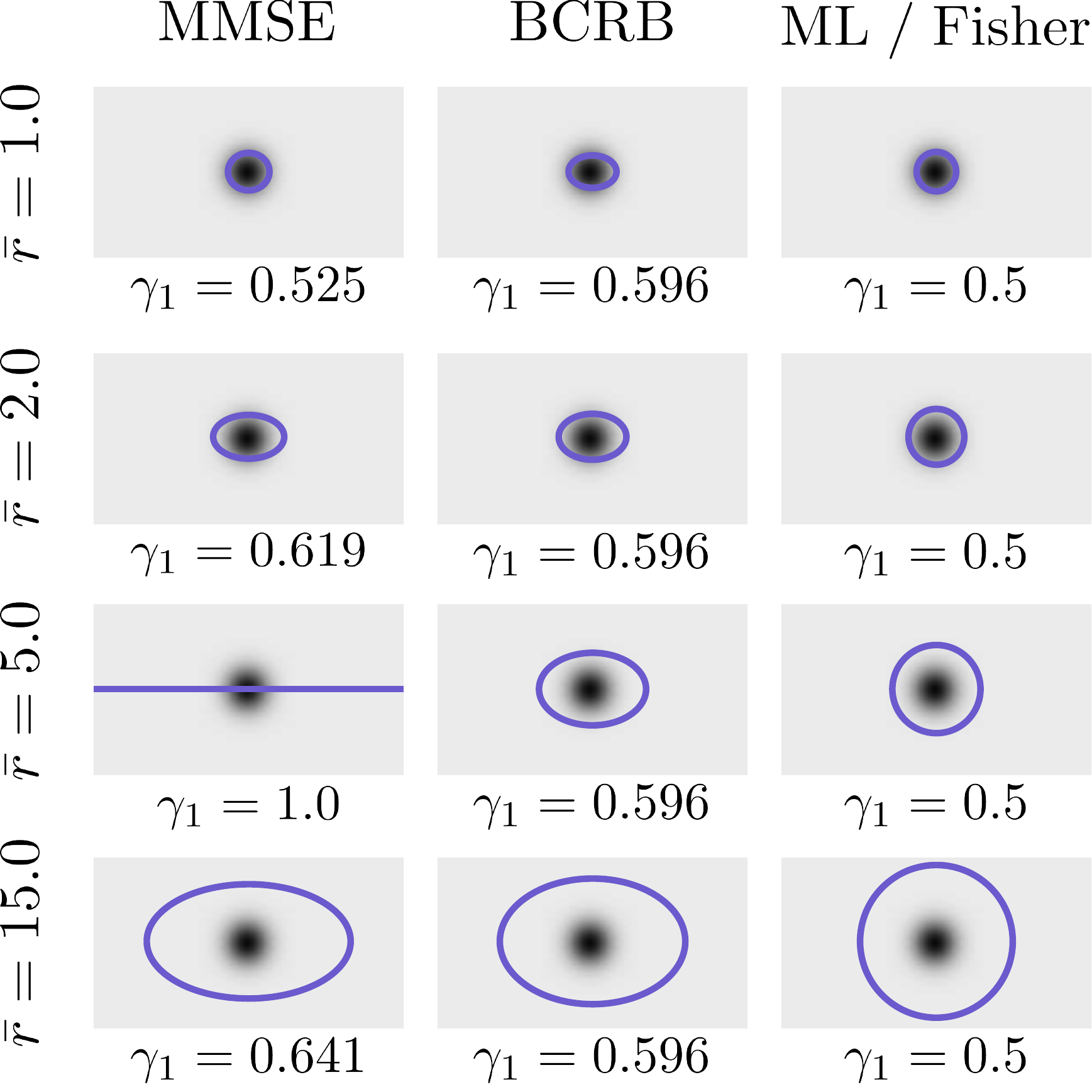}

}\caption{Encoding performance and optimal encoding as a function of population
rate constraint $\bar{r}$ and tuning width ratio $\gamma_{1}$; see
text and Figure \ref{fig:encoding-time} for details. Parameters are
$T=1,2\pi\bar{h}=1.1$, $\Sigma_{0}=\mathrm{diag}\left(1,1.05\right)$.}
\label{fig:encoding-rate}
\end{figure}

\begin{figure}
\includegraphics[width=0.5\columnwidth]{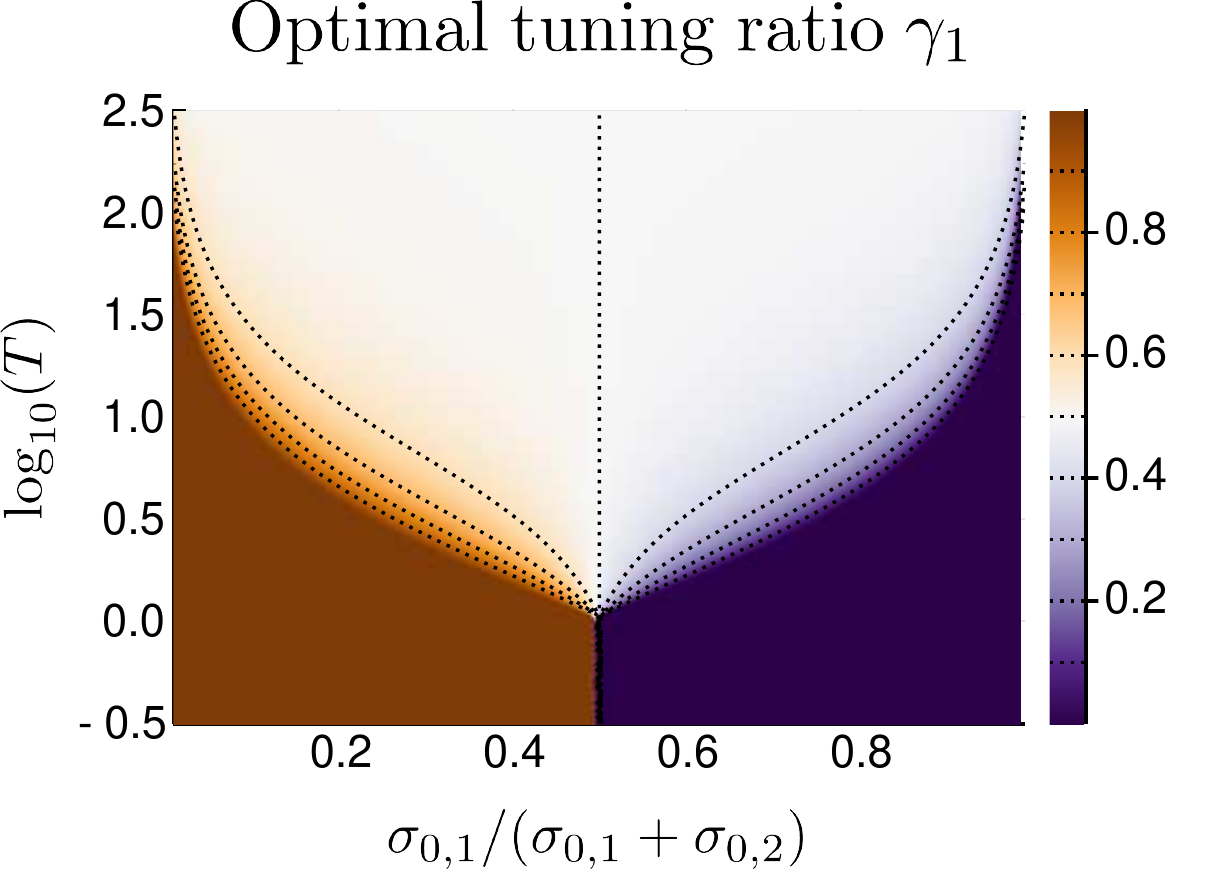}\hfill{}\includegraphics[width=0.5\columnwidth]{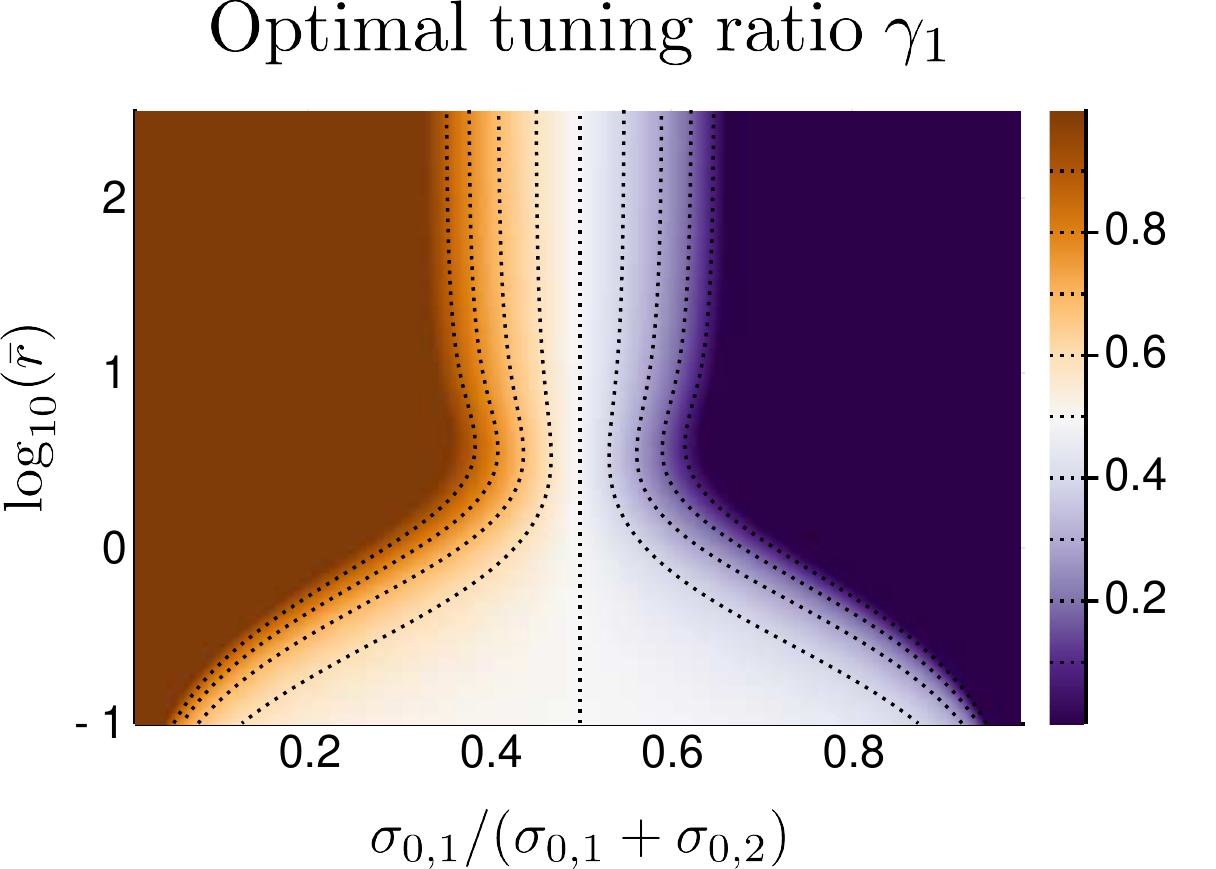}\caption{MMSE-optimal tuning width ratio $\gamma_{1}$ as a function of prior
ratio $\sigma_{0,1}/\left(\sigma_{0,1}+\sigma_{0,2}\right)$ and encoding
time (left) or population rate constraint (right), with the prior
variance determinant $\det\Sigma_{0}=\sigma_{0,1}^{2}\sigma_{0,2}^{2}$
fixed. Parameters on the left are $\bar{r}=10,2\pi\bar{h}=1$, and
on the right $T=1,2\pi\bar{h}=2$.}
\label{fig:encoding-opt}
\end{figure}

We compare our results to those obtained in \citep{Finkelstein2018}
for a model of neural encoding of head-direction in bats. The encoding
model in \citep{Finkelstein2018} includes neurons tuned to either
one or both of two angles -- azimuth or pitch. Preferred stimuli
are drawn from a uniform distribution on the resulting torus. The
estimation error of the ML estimator is compared in \emph{conjunctive
}(two-dimensional) coding, where each neuron encodes both directions,
and \emph{pure} (one-dimensional) coding, where each neuron encodes
one of the dimensions. Their analysis predicts that in a large population,
conjunctive coding always acheives lower error than pure coding, and
this effect is stronger for shorter decoding time. Our finding that
optimal coding is more symmetric for longer decoding time is \emph{the
opposite }of the conclusion in \citep{Finkelstein2018}. This discrepancy
may be attributed to the use of ML estimation as opposed to MMSE,
or to several differences between the models, as explored in more
detail in appendix \ref{sec:Pure-and-conjunctive}.

\section{Discussion}

We have studied MMSE-optimal multivariate encoding for infinite uniform
Gaussian neural population with population-level and neuron-level
rate constraints. This formulation allows for closed-form evaluation
of MMSE in the multivariate case (eqs. (\ref{eq:mmse-uniform-gaussian}),(\ref{eq:mmse-uniform-gaussian-diag}))
as well as simple lower and upper bounds (\ref{eq:mmse-bound}), and
the derivation of analytic results regarding optimal tuning. Specifically,
we have shown analytically that MMSE-optimal tuning functions in this
model are aligned with the principal components of the prior distribution,
and tuning is narrower in dimensions with larger prior variance. These
analytic results allow for a more computationally efficient numerical
optimization of encoding parameters, which we have used to evaluate
the two-dimensional case in more detail. We have found that encoding
only the dimension with higher variance is optimal for short encoding
times (see Figure \ref{fig:encoding-time}). We have also shown analytically
that MSE-optimal multivariate encoding in our model always involves
maximizing population firing rate, in contrast to the univariate case.
This suggest that the population-level constraint is especially relevant
in the multivariate setting. We have observed that optimal encoding
depends in a non-trivial way on this rate constraint (see Figure \ref{fig:encoding-rate}).

Furthermore, the analytic tractability of the approach allows us to
establish precise mathematical results, and to reach conceptually
novel and intuitively interpretable conclusions that may be applicable
to broader settings and to experimental data. Unfortunately, to the
best of our knowledge, there are currently few experimental results
on multivariate tuning properties that can be directly tested, but
our results make concrete predictions (under admittedly restricted
assumptions), when these become available.

Our framework differs from many previous studies in the direct optimization
of decoding MMSE rather than proxies such as Fisher information or
maximum likelihood estimation MSE, as well as in the explicit incorporation
of population-level rate constraints. We have found that optimization
of these proxies yields misleading results for short encoding times
-- as also observed in several previous works \citep{Bethge2002,YaeMei10,Pilarski2015}.
On the other hand, direct optimization of MMSE restricts the class
of models where the objective function can be readily evaluated, since
optimal decoding is generally intractable. To achieve mathematical
tractability in the multivariate case, we have focused on the case
of a Gaussian prior and Gaussian tuning functions that uniformly cover
an unbounded stimulus space. Note that in this energy-constrained
model, tuning width is also constrained, so that the effective space
of preferred stimuli which are likely to fire under the Gaussian prior
distribution of the state is bounded. Therefore, the unbounded space
of preferred stimuli used in our framework does not rule out its applicability
to biological settings with bounded stimulus space.

In the \emph{univariate }case with a single neuron, closed-form MMSE
has been previously derived in \citep{bethge_optimal_2003} for a
wide class of tuning functions, under the assumption of uniform prior.
Our work, in contrast, focuses on the case of uniform Gaussian tuning
and Gaussian prior. Although this is a more restricted class of tuning
functions, it facilitates analysis of the multivariate case. As far
as we are aware, this work is the first that derives closed-form decoding
MMSE for a multivariate stimulus. This suggests a possible direction
for future research in the convergence of these approaches to study
multivariate MMSE-optimal encoding for a wider class of tuning functions.

\section*{Acknowledgments}

We are grateful to Johnatan Aljadeff and Omri Barak for helpful discussions.
This work is partially supported by grant No. 451/17 from the Israel
Science Foundation, and by the Ollendorff center of the Viterbi Faculty
of Electrical Engineering at the Technion.

\subsection*{}

\printbibliography

\appendix

\section{Derivations}

\subsection{ML estimator\label{subsec:app-ML}}

We have defined a modified ML estimator (\ref{eq:modified-ml}), which
equals the prior mean when there are no spikes, since in this case
the standard ML estimator is undefined under uniform coding. The MSE
of this estimator may be computed directly from the law of total variance,

\begin{align*}
\E\left[\hat{X}_{t}^{\mathrm{ML}}|X,N_{t}=k\right] & =\begin{cases}
X & k\neq0\\
\mu_{0} & k=0
\end{cases}\\
\Var\left[\hat{X}_{t}^{\mathrm{ML}}|X,N_{t}=k\right] & =\frac{1}{k^{2}}\Var\left[\sum_{i=1}^{k}\theta_{i}|X,N_{t}=k\right]=\frac{1}{k}R^{-1}\quad\left(k\neq0\right)
\end{align*}
\begin{align}
 & \E\left[\left(X-\hat{X}_{t}^{\mathrm{ML}}\right)\left(X-\hat{X}_{t}^{\mathrm{ML}}\right)\transpose\Big|X\right]\nonumber \\
 & =\E\left[\Var\left[\hat{X}_{t}^{\mathrm{ML}}|X,N_{t}\right]|X\right]+\E\left[\left(\E\left[\hat{X}_{t}^{\mathrm{ML}}|X,N_{t}\right]-X\right)\left(\E\left[\hat{X}_{t}^{\mathrm{ML}}|X,N_{t}\right]-X\right)\transpose\Big|X\right]\nonumber \\
 & =\E\left[\frac{\mathbf{1}\left\{ N_{t}>0\right\} }{N_{t}}\Big|X\right]R^{-1}+\E\left[\mathbf{1}\left\{ N_{t}=0\right\} \left(\mu_{0}-X\right)\left(\mu_{0}-X\right)\transpose\Big|X\right]\nonumber \\
 & =e^{-rt}\left(\sum_{k=1}^{\infty}\frac{\left(rt\right)^{k}}{k!\,k}R^{-1}+\left(\mu_{0}-X\right)\left(\mu_{0}-X\right)\transpose\right)\label{eq:ml-cond-err}
\end{align}
where $r$ is the population rate given by (\ref{eq:total-rate}).
In the diagonal case (\ref{eq:diagonal}), we obtain the ML estimator's
conditional error by taking the trace of (\ref{eq:ml-cond-err}),
\[
\epsilon\left(\hat{X}_{t}^{\mathrm{ML}}\big|X\right)=e^{-rt}\left[\sum_{k=1}^{\infty}\frac{\left(rt\right)^{k}}{k!\,k}\sum_{i=1}^{m}\alpha_{i}^{2}+\left|\mu_{0}-X\right|^{2}\right],
\]
which yields the MSE (\ref{eq:ml-mse}) after taking the expected
value.

\subsection{Proofs \label{sec:proofs}}

\qdecreasing*   
\begin{proof}
Taking derivatives of (\ref{eq:q}) term by term yields the partial
derivatives of $q\left(s,r\right)$, 
\begin{align*}
q_{s}\left(s,r\right) & =re^{-r}\sum_{k=0}^{\infty}\frac{r^{k}}{k!}\frac{1}{\left(s+k+1\right)^{2}},\\
q_{r}\left(s,r\right) & =-se^{-r}\sum_{k=0}^{\infty}\frac{r^{k}}{k!}\frac{1}{\left(s+k\right)\left(s+k+1\right)},
\end{align*}
from which a tedious but straightforward calculation yields
\begin{align*}
\frac{d}{d\alpha}q\left(c_{1}\alpha^{2},c_{2}\alpha^{m}\right) & =e^{-r}\sum_{k=0}^{\infty}\frac{r^{k}}{k!}\frac{c_{1}c_{2}\alpha^{m+1}}{\left(s+k+1\right)^{2}}\left(2-m-\frac{m}{s+k}\right),
\end{align*}
where $s=c_{1}\alpha^{2},r=c_{2}\alpha^{m}$. This derivative is clearly
negative for positive $\alpha,c_{1},c_{2}$ and $m\geq2$.
\end{proof}
\trInv*	

As noted in section \ref{subsec:MMSE-optimal-multivariate-tuning},
applying this result to the prior precision $Q=\Sigma_{0}^{-1}$,
the optimal tuning shape matrix $R=U\Lambda U\transpose$ is a diagonal
matrix with its diagonal elements in decreasing order. That is, tuning
is aligned to the principal components, with narrower tuning (larger
diagonal elements of $R$) for components with larger prior variance
(smaller diagonal elements of $Q$).
\begin{proof}
We first prove the claim for the case of distinct eigenvalues: $\lambda_{1}<\lambda_{2}<\cdots\lambda_{n}$
and $q_{1}<q_{2}<\cdots q_{n}$. We apply Lagrange multipliers to
find necessary conditions on $U$. Using the fact that $\left(Q+U\Lambda U\transpose\right)^{-1}$
is symmetric and the relations 
\begin{align*}
\frac{\partial}{\partial X_{ij}}\mathrm{tr}\left[\left(Q+X\right)^{-1}\right] & =-\left[\left(Q+X\right)^{-2}\right]_{ji}\\
\frac{\partial}{\partial U_{ij}}\left(UBU\transpose\right)_{kl} & =\left(J^{kl}UB\transpose+J^{lk}UB\right)_{ij}
\end{align*}
(e.g., \citealt[eqs. (64),(79)]{Petersen2004matrix}), where $J^{kl}$
is the single-entry matrix $J_{ij}^{kl}\coloneqq\delta_{ik}\delta_{jl}$,
we obtain
\[
\frac{\partial}{\partial u_{ij}}\mathrm{tr}\left[\left(Q+U\Lambda U\transpose\right)^{-1}\right]=-2\left[\left(Q+U\Lambda U\transpose\right)^{-2}U\Lambda\right]_{ij}
\]
The optimization constraints are $\sum_{s}u_{ks}u_{ls}=\delta_{kl}$
for $k,l\in\left\{ 1,\ldots n\right\} $. Differentiation of the $(k,l)$
constraint with respect to $u_{ij}$ yields

\[
\frac{\partial}{\partial u_{ij}}\left[\sum_{s}u_{ks}u_{ls}-\delta_{kl}\right]=\delta_{ik}u_{lj}+\delta_{il}u_{kj}=\left[J^{kl}U+J^{lk}U\right]_{ij}
\]
leading to the necessary condition
\[
0=-2\left(Q+U\Lambda U\transpose\right)^{-2}U\Lambda+\sum_{kl}\mu_{kl}\left(J^{kl}+J^{lk}\right)U,
\]
where $\mu_{kl}$ are Lagrange multipliers. Multiplying on the right
by $U\transpose$ and using the constraint $UU\transpose=I$ yields
\[
\left(Q+U\Lambda U\transpose\right)^{-2}U\Lambda U\transpose=\frac{1}{2}\sum_{kl}\mu_{kl}\left(J^{kl}+J^{lk}\right).
\]
The right-hand side is symmetric, therefore so is the left-hand side.

Let $R=U\Lambda U\transpose$. We have found that $\left(Q+R\right)^{-2}R$
is symmetric, therefore $R$ commutes with $\left(Q+R\right)^{-2}$,
and they share an orthogonal basis of eigenvectors. Since $R$'s eigenvalues
are distinct, this basis is unique up to sign changes, so $(Q+R)^{-2}$
is also diagonalized by $U$, meaning $\left(Q+R\right)^{-2}=U\tilde{\Lambda}U\transpose$
for some diagonal $\tilde{\Lambda}$. Now,
\begin{align*}
RQ & =U\Lambda\left(\tilde{\Lambda}^{-1/2}-\Lambda\right)U\transpose=U\left(\tilde{\Lambda}^{-1/2}-\Lambda\right)\Lambda U\transpose=QR
\end{align*}
so
\[
q_{i}r_{ij}=\left(QR\right)_{ij}=\left(RQ\right)_{ij}=r_{ij}q_{j}.
\]
and since $\left\{ q_{i}\right\} $ are distinct and non-zero, $R$
is diagonal, $R=\mathrm{diag}\left(r_{1},\ldots r_{n}\right)$.

Since $R$ is diagonal and shares $\Lambda$'s eigenvalues, its diagonal
is a permutation of $\Lambda$'s diagonal, i.e., $r_{i}=\lambda_{\pi\left(i\right)}$
for some permutation $\pi$ on $\left\{ 1,\ldots n\right\} $. Thus
\[
r_{ij}=\delta_{ij}\lambda_{\pi\left(i\right)}=\sum_{k}\delta_{k,\pi\left(j\right)}\delta_{k,\pi\left(i\right)}\lambda_{k}=\left(P_{\pi}\Lambda P_{\pi}\transpose\right)_{ij}
\]
where $P_{\pi}$ is the permutation matrix $\left(P_{\pi}\right)_{ij}=\delta_{\pi\left(i\right),j}$.
Since $U$ affects the objective function only through $R=U\Lambda U\transpose$,
any $U$ satisfying $U\Lambda U\transpose=P_{\pi}\Lambda P_{\pi}\transpose$
is optimal, and in particular, so is $U=P_{\pi}$ (and other optima
are obtained by changing signs in $U$, $u_{ij}=\pm\left(P_{\pi}\right)_{ij}$).

The objective at the optimum is therefore of the form
\[
\mathrm{tr}\left[\left(Q+R\right)^{-1}\right]=\sum_{i}\left(q_{i}+r_{i}\right)^{-1}=\sum_{i}\left(q_{i}+\lambda_{\pi\left(i\right)}\right)^{-1}.
\]
To show that this is minimized by the inversion permutation $\pi\left(i\right)=n+1-i$,
assume to the contrary that there are $i<j$ such that $\pi\left(i\right)<\pi\left(j\right)$.
Denote by $\pi'$ the permutation obtained from $\pi$ by switching
these two elements,
\[
\pi'\left(k\right)=\begin{cases}
\pi\left(j\right) & k=i\\
\pi\left(i\right) & k=j\\
\pi\left(k\right) & k\neq i,j
\end{cases}
\]
Since $i<j,\pi\left(i\right)<\pi\left(j\right)$ we have $q_{i}<q_{j}$
and $\lambda_{\pi\left(i\right)}<\lambda_{\pi\left(j\right)}$, therefore
\begin{multline*}
\sum_{k}\left(q_{k}+\lambda_{\pi\left(k\right)}\right)^{-1}-\sum_{k}\left(q_{k}+\lambda_{\pi'\left(k\right)}\right)^{-1}\\
=\left(\frac{1}{q_{i}+\lambda_{\pi\left(i\right)}}+\frac{1}{q_{j}+\lambda_{\pi\left(j\right)}}\right)-\left(\frac{1}{q_{i}+\lambda_{\pi\left(j\right)}}+\frac{1}{q_{j}+\lambda_{\pi\left(i\right)}}\right)\\
=\frac{q_{i}+q_{j}+\lambda_{\pi\left(i\right)}+\lambda_{\pi\left(j\right)}}{\left(q_{i}+\lambda_{\pi\left(i\right)}\right)\left(q_{j}+\lambda_{\pi\left(j\right)}\right)\left(q_{i}+\lambda_{\pi\left(j\right)}\right)\left(q_{j}+\lambda_{\pi\left(i\right)}\right)}\left(q_{j}-q_{i}\right)\left(\lambda_{\pi\left(j\right)}-\lambda_{\pi\left(i\right)}\right)>0,
\end{multline*}
contradicting the assumption that $P_{\pi}$ is optimal.

This concludes the proof for the case where all eigenvalues are distinct.
The result is easily extended to the case of multiple eigenvalues
by the continuity of the objective function: Denote the objective
function for eigenvalues $\boldsymbol{q}=(q_{i})_{i=1}^{n},\boldsymbol{\lambda}=(\lambda_{i})_{i=1}^{n}$
by 
\[
f_{\boldsymbol{q},\boldsymbol{\lambda}}\left(U\right)\coloneqq\mathrm{tr}\left[\left(\mathrm{diag}\left(\boldsymbol{q}\right)+U\mathrm{diag}\left(\boldsymbol{\lambda}\right)U\transpose\right)^{-1}\right].
\]
Let $U^{*}$ denote the anti-diagonal $U$, which is optimal when
eignevalues are distinct, and assume $U^{*}$ is not optimal for some
\textbf{$\boldsymbol{q}_{0},\boldsymbol{\lambda}_{0}$} (with possibly
non-distinct elements), i.e., there is an orthogonal matrix $U'$
such that $f_{\boldsymbol{q}_{0},\boldsymbol{\lambda}_{0}}\left(U'\right)<f_{\boldsymbol{q}_{0},\boldsymbol{\lambda}_{0}}\left(U^{*}\right)$.
Since $f$ depends continuously on \textbf{$\boldsymbol{q},\boldsymbol{\lambda}$},
this inequality holds in a neighborhood of \textbf{$\boldsymbol{q}_{0},\boldsymbol{\lambda}_{0}$},
which includes \textbf{$\boldsymbol{q},\boldsymbol{\lambda}$ }with
distinct elements, contradicting the optimality result above.
\end{proof}

\subsection{Fisher information of an infinite population\label{sec:Fisher-infinite}}

We derive Fisher information for an infinite neural population described
by a marked point process with rate-density $\lambda\left(x;\theta\right)$,
without assumptions on the form of $\lambda$, so the population might
have non-Gaussian tuning and might be non-uniform.

Given an infinite population with rate-density $\lambda\left(x;\theta\right)$,
the likelihood of a spike sequence $\left(t_{k},\theta_{k}\right)_{k=1}^{N_{T}}$
is given by 
\[
p\left(\left(t_{k},\theta_{k}\right)_{k=1}^{N_{T}}|X=x\right)=e^{-r\left(x\right)T}\prod_{k=1}^{N_{T}}\lambda\left(x;\theta_{k}\right),
\]
where $r\left(x\right)\coloneqq\int\lambda\left(x;\theta\right)d\theta$
is the total population rate in response to $x$ (equation (\ref{eq:likelihood})
is the special case of a uniform population, where the population
rate $r\left(x\right)$ is constant). Denoting the log-likelihood
by $L_{t}\left(x\right)$, 
\begin{align*}
L_{t}\left(x\right) & \coloneqq\log p\left(\left(t_{k},\theta_{k}\right)_{k=1}^{N_{T}}|X=x\right)=-r\left(x\right)T+\sum_{k=1}^{N_{T}}\log\lambda\left(x;\theta_{k}\right)\\
 & =-r\left(x\right)T+\int_{0}^{T}\int_{\mathbb{R}^{m}}\log\lambda\left(x;\theta\right)N\left(dt,d\theta\right).
\end{align*}
Using the shorthand $\partial_{i}\coloneqq\partial/\partial x_{i}$,
Fisher information is given by 
\begin{align*}
J_{ij}\left(x\right) & =-\E\left[\partial_{i}\partial_{j}L_{t}\left(x\right)|X=x\right]\\
 & =-\E\left[-\partial_{i}\partial_{j}r\left(x\right)+\int_{0}^{T}\int_{\mathbb{R}^{m}}\left(\partial_{i}\partial_{j}\log\lambda\left(x;\theta\right)\right)N\left(dt,d\theta\right)|X=x\right]\\
 & =T\partial_{i}\partial_{j}r\left(x\right)-T\int\left(\partial_{i}\partial_{j}\log\lambda\left(x;\theta\right)\right)\lambda\left(x;\theta\right)d\theta\\
 & =T\partial_{i}\partial_{j}r\left(x\right)-T\int\left(\partial_{i}\partial_{j}\lambda\left(x;\theta\right)-\frac{\left(\partial_{i}\lambda\left(x;\theta\right)\right)\left(\partial_{j}\lambda\left(x;\theta\right)\right)}{\lambda\left(x;\theta\right)}\right)d\theta\\
 & =T\int\frac{\left(\partial_{i}\lambda\left(x;\theta\right)\right)\left(\partial_{j}\lambda\left(x;\theta\right)\right)}{\lambda\left(x;\theta\right)}d\theta.
\end{align*}

\section{Multiple sub-populations}

To allow a more direct comparison to the analysis of \citep{Finkelstein2018},
we extend our model in two ways. First, we consider $M$ uniform sub-population,
each with its own maximal rate-density $h_{i}$ and tuning shape $R_{i}$.
Second, the tuning in the $i$th sub-population takes the form 
\begin{equation}
\lambda_{i}\left(x;\theta\right)\coloneqq h_{i}\exp\left(-\frac{1}{2}\left\Vert H_{i}x-\theta\right\Vert _{R_{i}}^{2}\right),\label{eq:gauss-tc-H}
\end{equation}
where the state dimensionality is now denoted $n$, and $H_{i}\in\mathbb{R}^{m_{i}\times n}$
is a matrix of full row rank ($m_{i}\leq n$), so that the sub-population
encodes the linear projection $H_{i}x\in\mathbb{R}^{m_{i}}$ of the
$n$-dimensional state to an $m_{i}$-dimensional subspace.

The derivation of the closed-form MMSE (\ref{eq:mmse-uniform-gaussian})
can be readily extended to this case. The posterior Gaussian distribution
has mean and variance\footnote{the formulation of \citep[Theorem 1]{Snyder1977}, from which we have
deduced (\ref{eq:posterior}) above, includes a projection matrix
$H$. The extension to multiple sub-populations is straightforward.}
\begin{align}
\mu_{T} & =\Sigma_{T}\left(\Sigma_{0}^{-1}\mu_{0}+\sum_{j=1}^{M}H_{j}\transpose R_{j}H_{j}\sum_{i=1}^{N_{T}^{\left(j\right)}}\theta_{i}^{\left(j\right)}\right),\label{eq:post-mean-subpops}\\
\Sigma_{T} & =\left(\Sigma_{0}^{-1}+\sum_{j=1}^{M}N_{T}^{\left(j\right)}H_{j}\transpose R_{j}H_{j}\right)^{-1},\label{eq:post-precision-subpops}
\end{align}
where $\theta_{i}^{\left(j\right)}$ is the preferred stimulus of
the neuron firing the $i$-th spike from that population, and $N_{T}^{\left(j\right)}$
the number of spikes observed from the $j$-th population in the time
interval $\left[0,T\right]$. Since $N_{T}^{\left(j\right)}$ are
independent Poisson variables with rate
\[
r_{j}=\sqrt{\frac{\left(2\pi\right)^{m_{j}}}{\det R_{j}}}h_{j},
\]
the MMSE takes the form
\begin{align}
\epsilon_{\mathrm{MMSE}} & =\E\,\mathrm{tr}\left[\left(\Sigma_{0}^{-1}+\sum_{j=1}^{M}N_{T}^{\left(j\right)}H_{j}\transpose R_{j}H_{j}\right)^{-1}\right]\label{eq:mmse-uniform-gaussian-mixture}\\
 & =\exp\left(-\sum_{j=1}^{M}r_{j}T\right)\sum_{k_{1},\ldots k_{M}=0}^{\infty}\left[\prod_{j=1}^{M}\frac{\left(r_{j}T\right)^{k_{j}}}{k_{j}!}\right]\mathrm{tr}\left(\left(\Sigma_{0}^{-1}+\sum_{j=1}^{M}k_{j}H_{j}\transpose R_{j}H_{j}\right){}^{-1}\right).
\end{align}

Fisher information in the $i$-th subpopulation is given by

\begin{equation}
J_{i}\left(x\right)=T\int\frac{\nabla_{x}\lambda_{i}\left(x;\theta\right)\nabla_{x}\lambda_{i}\left(x;\theta\right)\transpose}{\lambda_{i}\left(x;\theta\right)}d\theta=h'_{i}T\left|R_{i}\right|^{-1/2}H_{i}\transpose R_{i}H_{i}=r_{i}TH_{i}\transpose R_{i}H_{i}.\label{eq:uniform-gaussian-fisher-H}
\end{equation}
Since the subpopulations fire independently given the state $X$,
Fisher information of the entire population is given by the sum $J\left(x\right)=\sum_{i}J_{i}\left(x\right)$.

For comparison with \citep{Finkelstein2018}, we focus on the case
of diagonal prior variance, with each sub-population encoding one
of the dimensions\footnote{With the inclusion of the projection matrices $H_{j}$, which may
not align with principal components of the prior distribution, we
can no longer reduce the problem in general to the diagonal cases
of tuning aligned with principal components, as in Claim \ref{claim:tr-inv}.
Our motivation for focusing on the diagonal case, with projections
aligned to the axes, is allowing comparison to the results of \citet{Finkelstein2018}.
Note that in \citep{Finkelstein2018}, the prior is uniform, so that
all directions are principal components.}. Specifically, assume the prior variance is diagonal (\ref{eq:diagonal-prior}),
and the $j$-th subpopulation encodes only the $j$-th dimension,
with tuning width $\alpha_{j}$, i.e., $H_{j}\in\mathbb{R}^{1\times n}$
is a row vector with
\begin{equation}
\left[H_{j}\right]_{1i}=\delta_{ij},\quad R_{j}=\alpha_{j}^{-2}.\label{eq:pure-1d}
\end{equation}
The firing rate of the $j$-th subpopulation is 
\[
r_{j}=\sqrt{2\pi}\alpha h_{j},
\]
and the MMSE is given by
\begin{align}
\epsilon_{\mathrm{MMSE}} & =\E\,\mathrm{tr}\left[\left(\Sigma_{0}^{-1}+\mathrm{diag}\left(\alpha_{1}^{-2}N_{T}^{\left(1\right)},\ldots,\alpha_{M}^{-2}N_{T}^{\left(M\right)}\right)\right){}^{-1}\right]\nonumber \\
 & =\E\sum_{j=1}^{M}\frac{1}{\sigma_{0,j}^{-2}+\alpha_{j}^{-2}N_{T}^{\left(j\right)}}\nonumber \\
 & =\sum_{j=1}^{M}\sigma_{0,j}^{2}q\left(\frac{\alpha_{j}^{2}}{\sigma_{0,j}^{2}},r_{j}\right)\label{eq:mmse-uniform-gaussian-mixture-pure}
\end{align}
where $q$ is defined in (\ref{eq:q}). Fisher information (\ref{eq:uniform-gaussian-fisher-H})
under the setting (\ref{eq:pure-1d}) is given by 
\begin{align*}
J\left(x\right) & =\sum_{i=1}^{M}J_{i}\left(x\right)=\sum_{i=1}^{M}\frac{h'_{i}T}{\alpha_{i}}H_{i}\transpose H_{i}\\
 & =\diag\left(\frac{h_{1}'T}{\alpha_{1}},\ldots\frac{h_{M}'T}{\alpha_{M}}\right)=\diag\left(\frac{r_{1}T}{\alpha_{1}^{2}},\ldots\frac{r_{M}T}{\alpha_{M}^{2}}\right).
\end{align*}

The ML estimator is generally not unique for the model (\ref{eq:gauss-tc-H})
incorporating the projection matrix $H$. However, in the specific
case (\ref{eq:pure-1d}), it is uniquely defined when there is at
least one spike from each subpopulation. Analogously to the modified
ML estimator for a single population defined in section \ref{subsec:ML},
we define a modified ML estimator $\hat{X}_{t}^{\mathrm{ML}}=\left(\hat{X}_{t}^{1;\mathrm{ML}},\hat{X}_{t}^{2;\mathrm{ML}}\right),$
where the $i$-th dimension is estimated as

\begin{equation}
\hat{X}_{t}^{i;\mathrm{ML}}=\begin{cases}
\frac{1}{N_{t}^{\left(i\right)}}\sum_{k=1}^{N_{t}^{\left(i\right)}}\theta_{k}^{\left(i\right)} & N_{t}^{\left(i\right)}>0,\\
\mu_{0,i} & N_{t}^{\left(i\right)}=0.
\end{cases}\label{eq:modified-ml-subpops}
\end{equation}
As in the single population case, this modified ML estimator may be
obtained from the MMSE estimator (\ref{eq:post-mean-subpops}) in
the limit $\Sigma_{0}^{-1}\to0$, and its MSE may be computed from
the law of total variance,
\begin{align*}
\E\left[\hat{X}_{t}^{i;\mathrm{ML}}|X,N_{t}\right] & =\begin{cases}
X^{\left(i\right)}, & N_{t}^{\left(i\right)}>0,\\
\mu_{0,i}, & N_{t}^{\left(i\right)}=0,
\end{cases}\\
\mathrm{Var}\left[\hat{X}_{t}^{i;\mathrm{ML}}|X,N_{t}\right] & =\begin{cases}
\alpha_{i}^{2}/N_{t}^{\left(i\right)}, & N_{t}^{\left(i\right)}>0,\\
0, & N_{t}^{\left(i\right)}=0.
\end{cases}
\end{align*}
\begin{align}
 & \E\left[\left(X^{\left(i\right)}-\hat{X}_{t}^{i;\mathrm{ML}}\right)^{2}\Big|X\right]=\E\left[\Var\left[\hat{X}_{t}^{i;\mathrm{ML}}|X,N_{t}\right]|X\right]+\E\left[\left(\E\left[\hat{X}_{t}^{i;\mathrm{ML}}|X,N_{t}\right]-X^{\left(i\right)}\right)^{2}\Big|X\right]\nonumber \\
 & =\E\left[\frac{\mathbf{1}\left\{ N_{t}^{\left(i\right)}>0\right\} }{N_{t}^{\left(i\right)}}\Big|X\right]\alpha_{i}^{2}+\E\left[\mathbf{1}\left\{ N_{t}^{\left(i\right)}=0\right\} \left(\mu_{0,i}-X^{\left(i\right)}\right)^{2}\Big|X\right]\nonumber \\
 & =e^{-r_{i}t}\left(\sum_{k=1}^{\infty}\frac{\left(r_{i}t\right)^{k}}{k!\,k}\alpha_{i}^{2}+\left(\mu_{0,i}-X^{\left(i\right)}\right)^{2}\right).\label{eq:ml-cond-err-1}
\end{align}
Taking the expected value and summing over $i$ yields the MSE
\begin{equation}
\epsilon\left(\hat{X}_{t}^{\mathrm{ML}}\right)=\sum_{i=1}^{M}e^{-r_{i}t}\left[\left(\sum_{k=1}^{\infty}\frac{\left(r_{i}t\right)^{k}}{k!\,k}\right)\alpha_{i}^{2}+\sigma_{i}^{2}\right].\label{eq:ml-mse-subpops}
\end{equation}

\section{Pure and conjunctive coding of two-dimensional stimuli\label{sec:Pure-and-conjunctive}}

We compare our results to those obtained in \citep{Finkelstein2018}
for a model of neural encoding of head-direction in bats. As noted
in section \ref{subsec:Two-dimensional-uniform-Gaussian}, the encoding
model in \citep{Finkelstein2018} includes neurons tuned to either
one or both of two angles -- azimuth or pitch. Preferred stimuli
are drawn from a uniform distribution on the resulting torus. The
estimation error of the ML estimator is compared in \emph{conjunctive
}(two-dimensional) coding, where each neuron encodes both directions,
and \emph{pure} (one-dimensional) coding, where each neuron encodes
one of the dimensions. Their analysis predicts that in a large population,
conjunctive coding always acheives lower error than pure coding, and
this effect is stronger for shorter decoding time. This is in contrast
to our finding that optimal coding is more symmetric for longer decoding
time. This difference may be attributed to the use of ML estimation
rather than MMSE estimation, or to several other differences between
the models, as explained below.

To explore the difference between our results and those of \citep{Finkelstein2018},
we focus on the case of \emph{symmetric} prior and tuning, and compare
the MMSE of conjunctive and pure encoding. Specifically, we compute
the ratio of MMSE in conjunctive coding and MMSE in pure coding. We
consider several variations of our model to bring it closer to that
of \citep{Finkelstein2018}. Our model, as defined in the main text
and studied in the two-dimensional case in section \ref{subsec:Two-dimensional-uniform-Gaussian},
differs from \citep{Finkelstein2018} in several ways:
\begin{description}
\item [{Estimator}] Optimal MMSE estimation vs. (suboptimal) ML esimation
\item [{State~space}] Encoding of a state in Euclidean space vs. encoding
of angles
\item [{Number~of~subpopulations}] A single population vs. two pure sub-populations
\item [{Subpopulation~dimensionality}] In our model, populations have
preferred stimuli covering $\mathbb{R}^{2}$, but tuning curves may
be arbitrarily thin and long. In \citep{Finkelstein2018}, each \emph{pure
}populations has preferred stimuli covering a 1-d space and finite
tuning width.
\item [{Normalization}] Based on the derivation from firing rates constraints,
we have considered populations with fixed firing rates $h,r$ and
varying tuning widths. The model in \citep{Finkelstein2018} fixes
the total rate $r$ and the tuning width $\alpha$ for pure and conjunctive
encoding, resulting in different maximal firing rates per neuron between
the two cases.
\end{description}
Since our framework allows computation of MMSE only for Euclidean
state space, we do not address the problem of encoding of angles here.
We therefore consider several possible definitions of the \emph{pure
}pouplation:
\begin{itemize}
\item A single pure population encoding one of the dimensions vs. two pure
subpopulations each encoding one dimension
\item The pure subpopulations are defined in one of three ways:
\begin{enumerate}
\item 2-d pure populations ($\theta\in\mathbb{R}^{2}$) with same firing
rates $h,r$ as in conjunctive code. MMSE is evaluated in the infinitely
narrow limit $\alpha_{1}\to0,\alpha_{2}\to\infty$ or vice versa.
This is the setting studied in section \ref{subsec:Two-dimensional-uniform-Gaussian}.
\item 1-d pure populations ($\theta\in\mathbb{R}$) with same firing rates
$h,r$ as in conjunctive code. This setting allows finite population
rates with finite tuning width $\alpha$, which is dictated by the
rate constraints.
\item 1-d pure populations ($\theta\in\mathbb{R}$) with same width $\alpha$
and population rate $r$ as in conjunctive code. In this case, the
firing rate-density $h$ differs between pure and conjunctive coding.
\end{enumerate}
\end{itemize}
More precisely, in the last two cases, each sub-population has tuning
as given by (\ref{eq:gauss-tc-H}) and (\ref{eq:pure-1d}). When there
are two sub-populations, the firing rate $r$ that is equalized between
pure and conjunctive coding is the total rate population rate, so
that each sub-population has firing rate $r/2$. The errors for a
single pure sub-population are calculated from the same equations
((\ref{eq:mmse-uniform-gaussian-diag}) and (\ref{eq:mmse-uniform-gaussian-mixture-pure}))
by assigning rate $r$ to one sub-population and rate 0 to the other.

MMSE ratio for pure vs. conjunctive coding in each of these six models
as a function of encoding time is shown in Figure \ref{fig:pure-conj}.
Evidently, the preference for pure vs. conjunctive coding as a function
of encoding time depends on the details of the model. The model of
a single 2d pure population (top left) is the one appearing in the
main text. The model of two 1d sub-populations with the same tuning
width as the conjunctive population (bottom right) is the one most
similar to \citep{Finkelstein2018}. In both these cases, we find
that pure coding is relatively better for short decoding time, which
is the opposite of the conclusion in \citep{Finkelstein2018}.
\begin{center}
\begin{figure}
\begin{centering}
\begin{tabular}{>{\centering}m{1cm}>{\centering}m{0.4\columnwidth}>{\centering}m{0.4\columnwidth}}
 & 1 pure population & 2 pure populations\tabularnewline
\begin{turn}{90}
2d
\end{turn} & \includegraphics[width=0.4\columnwidth]{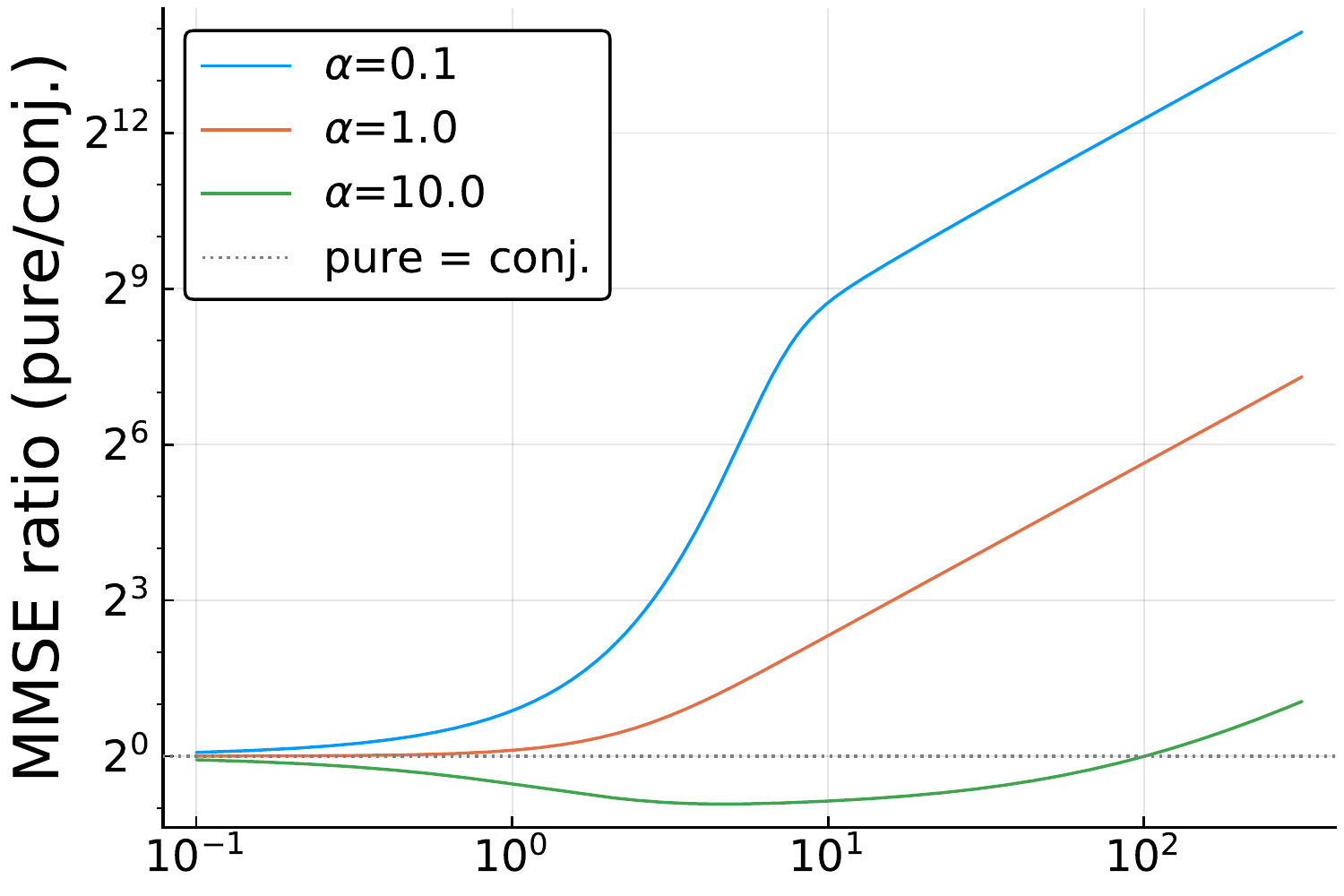} & \includegraphics[width=0.4\columnwidth]{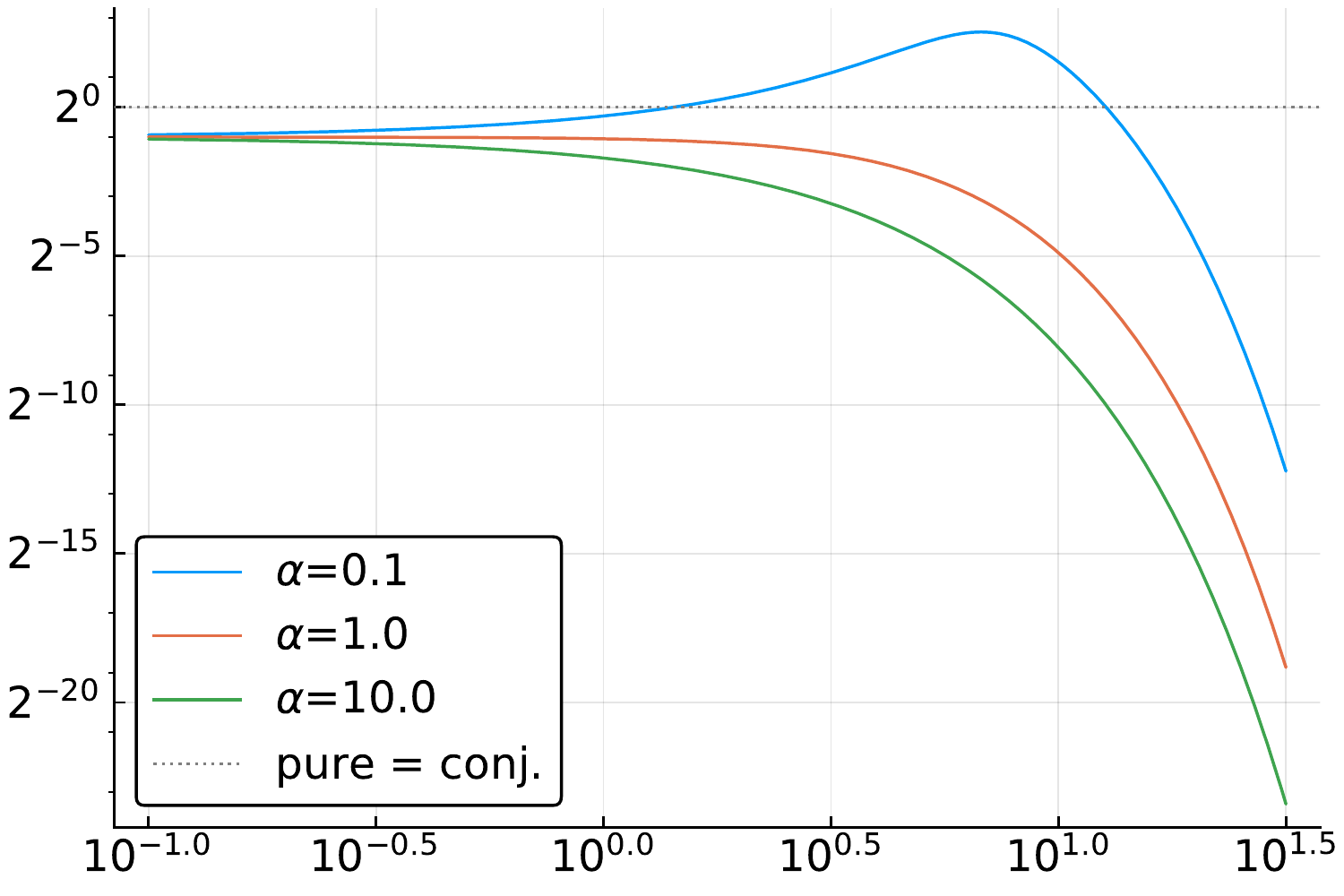}\tabularnewline
\begin{turn}{90}
1d, same rate-density $h$
\end{turn} & \includegraphics[width=0.4\columnwidth]{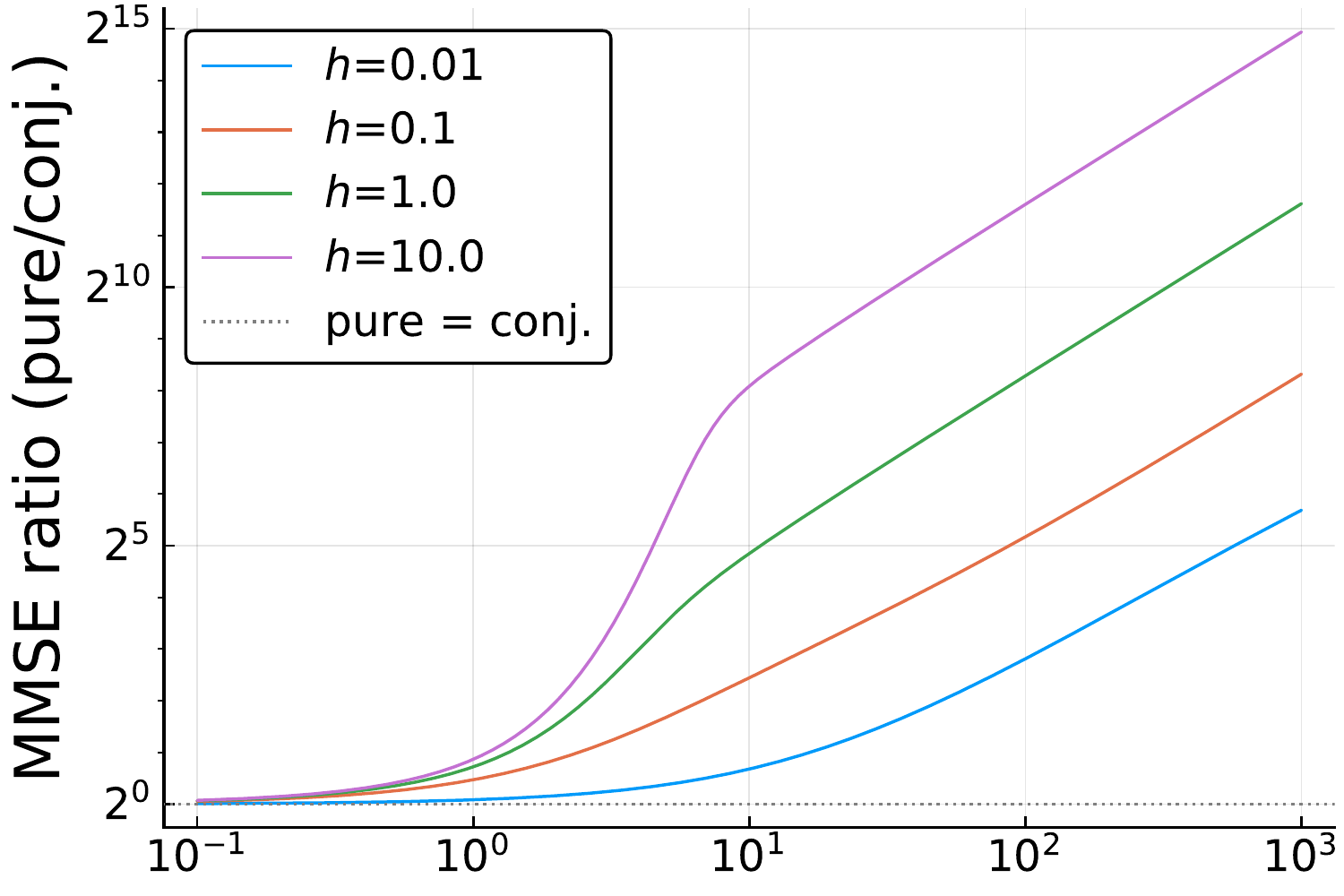} & \includegraphics[width=0.4\columnwidth]{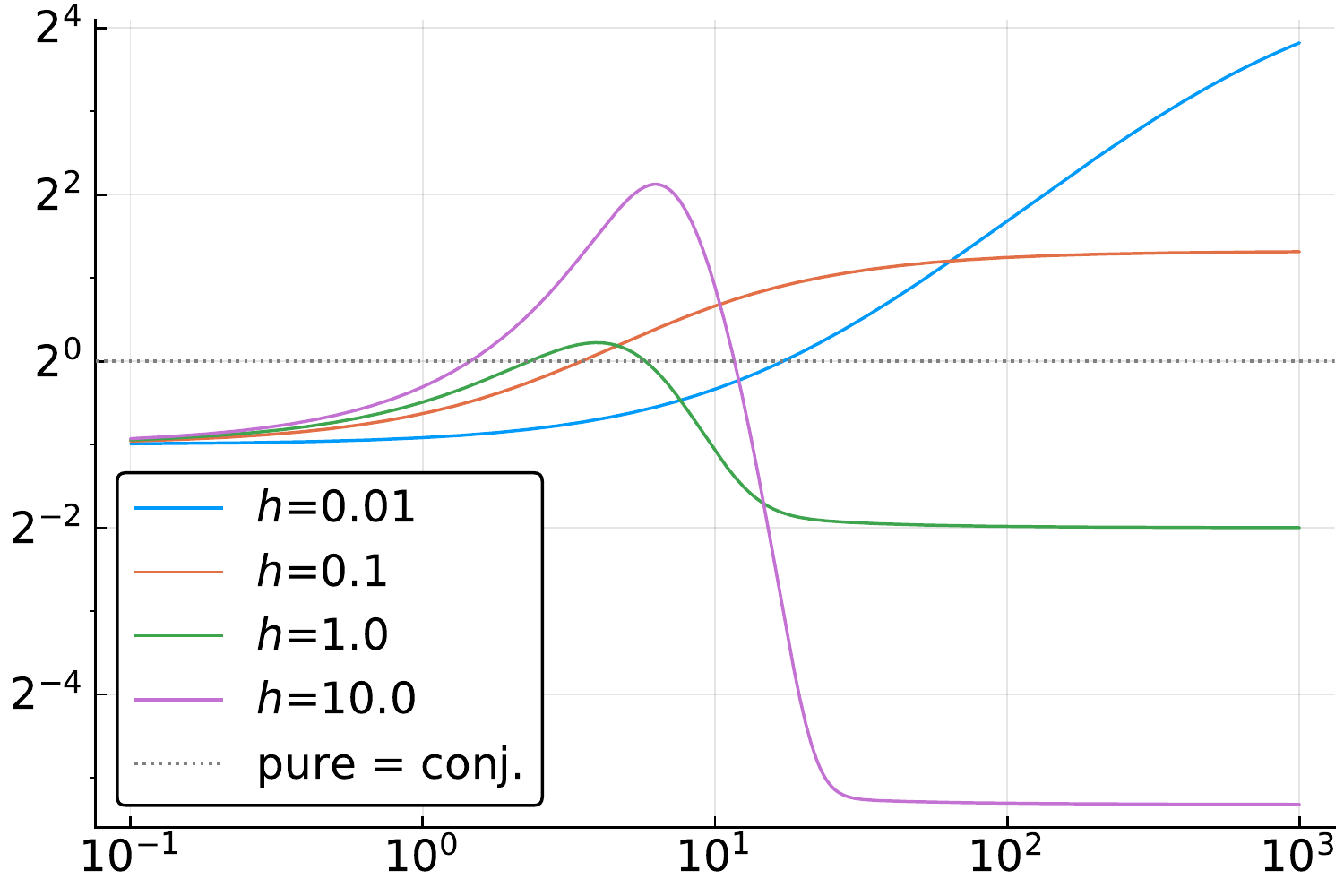}\tabularnewline
\begin{turn}{90}
1d, same width $\alpha$
\end{turn} & \includegraphics[width=0.4\columnwidth]{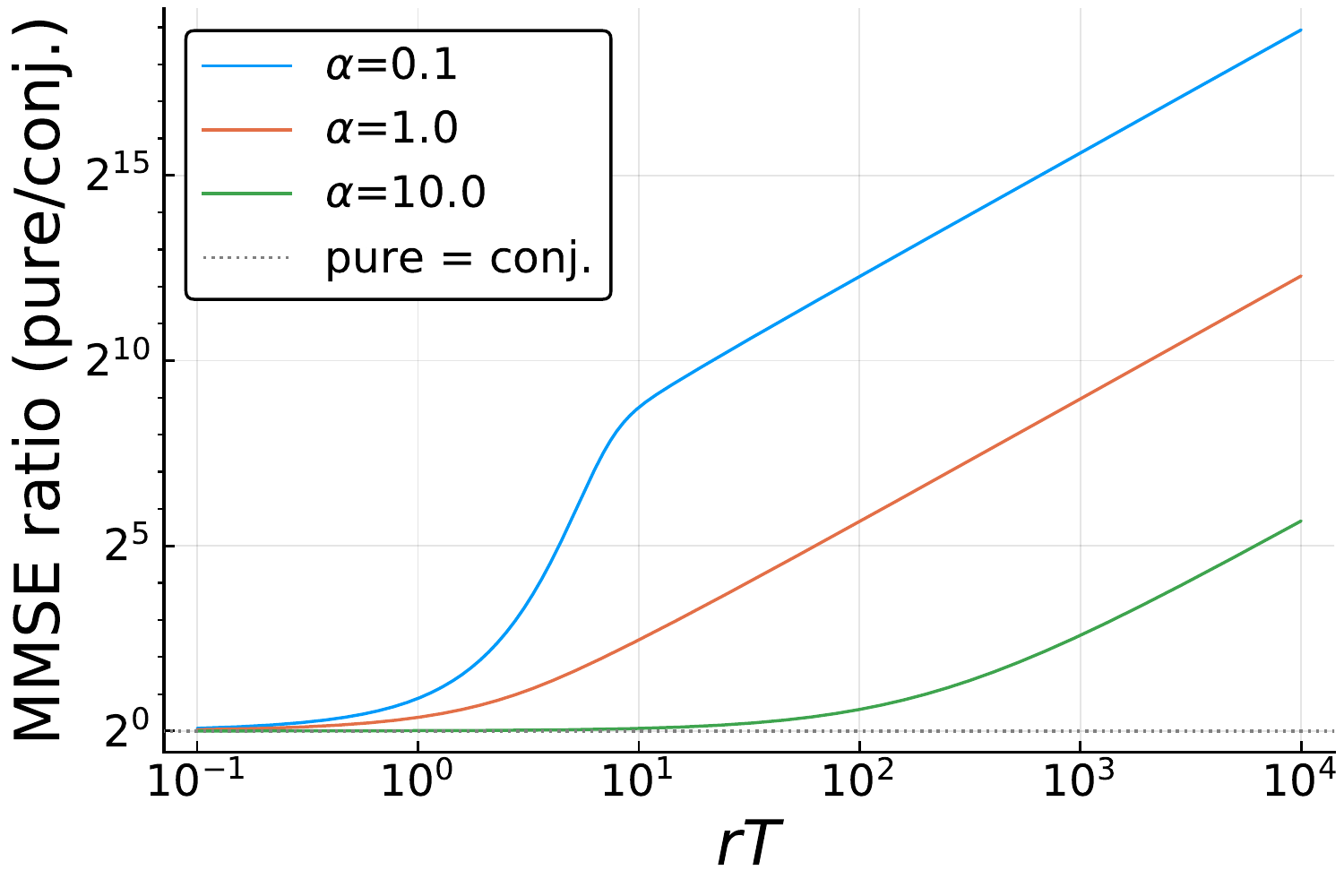} & \includegraphics[width=0.4\columnwidth]{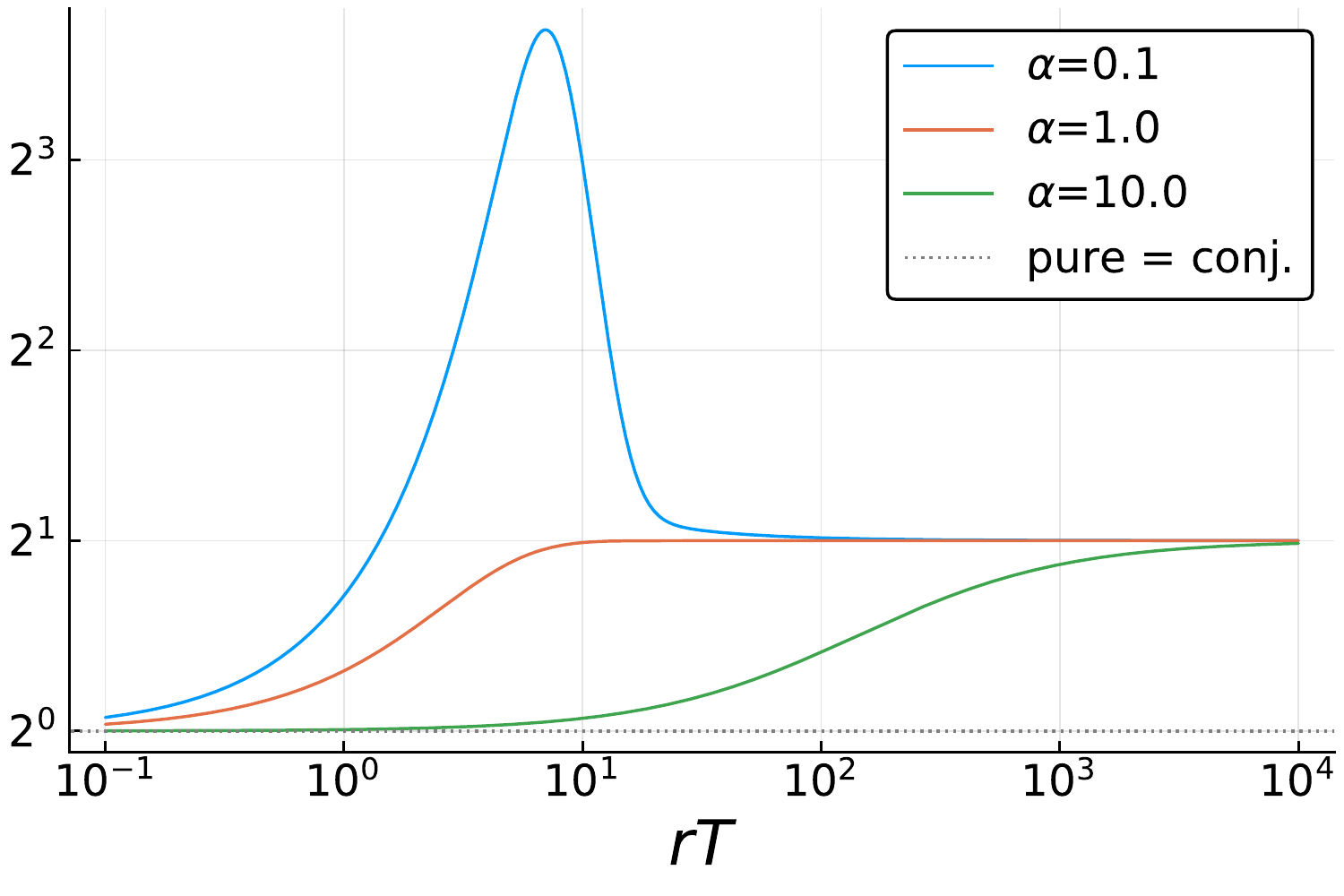}\tabularnewline
\end{tabular}
\par\end{centering}
\caption{MMSE ratio of pure and conjunctive coding for various models.\emph{
}See text for more details. In the first row, each pure population
is uniform in $\mathbb{R}^{2}$ with infinitely thin tuning. In the
second row, each pure population is uniform in $\mathbb{R}$ and rate-density
$h$ is the same for the pure and conjunctive codes. In the third
row, each pure population is uniform in $\mathbb{R}$ and tuning width
$\alpha$ is the same for the pure and conjunctive codes.}
\label{fig:pure-conj}
\end{figure}
\par\end{center}

To assess whether the different dependence on encoding time is related
to the use of ML estimation in \citep{Finkelstein2018}, we also evaluate
the ratio of ML MSE in pure and conjunctive coding for the models
involving 1d pure populations in Figure \ref{fig:pure-conj-ml}. The
model of a 2d pure population is omitted from this figure, since ML
estimation in this model has infinite variance (as seen by taking
the limit $\alpha_{1}\to0,\alpha_{2}\to\infty$ in (\ref{eq:ml-mse})
with $r\propto\alpha_{1}\alpha_{2}$ fixed). The results for ML estimation
are qualitatively similar to MMSE estimation in this symmetric case,
except in some cases where pure coding achieves lower ML MSE than
conjunctive coding, but worse MMSE (see, e.g., the bottom rows of
Figures \ref{fig:pure-conj} and \ref{fig:pure-conj-ml}). In particular,
pure coding remains preferable for short decoding time when using
ML estimation.
\begin{center}
\begin{figure}
\begin{centering}
\begin{tabular}{>{\centering}m{1cm}>{\centering}m{0.4\columnwidth}>{\centering}m{0.4\columnwidth}}
 & 1 pure population & 2 pure populations\tabularnewline
\begin{turn}{90}
1d, same rate-density $h$
\end{turn} & \includegraphics[width=0.4\columnwidth]{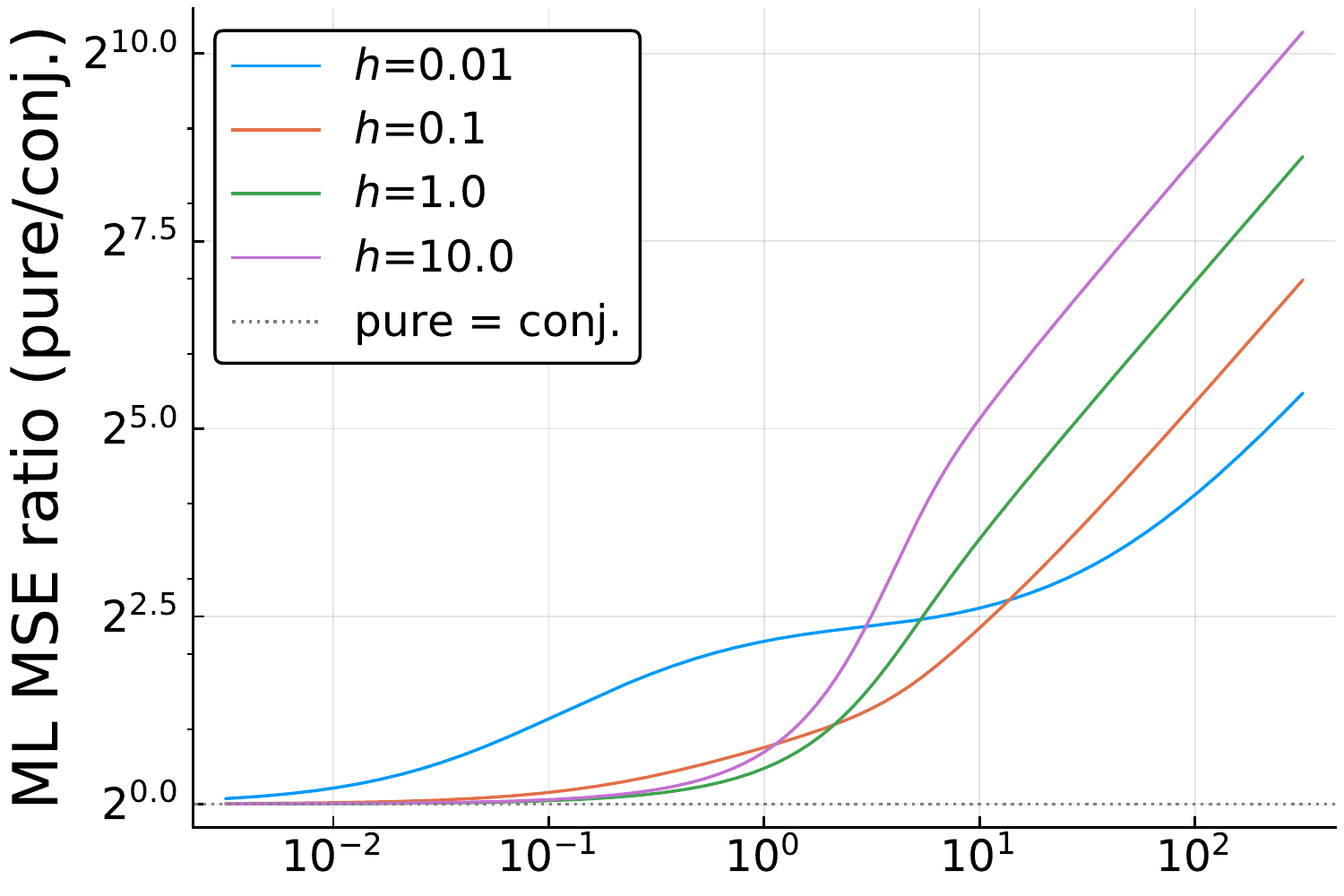} & \includegraphics[width=0.4\columnwidth]{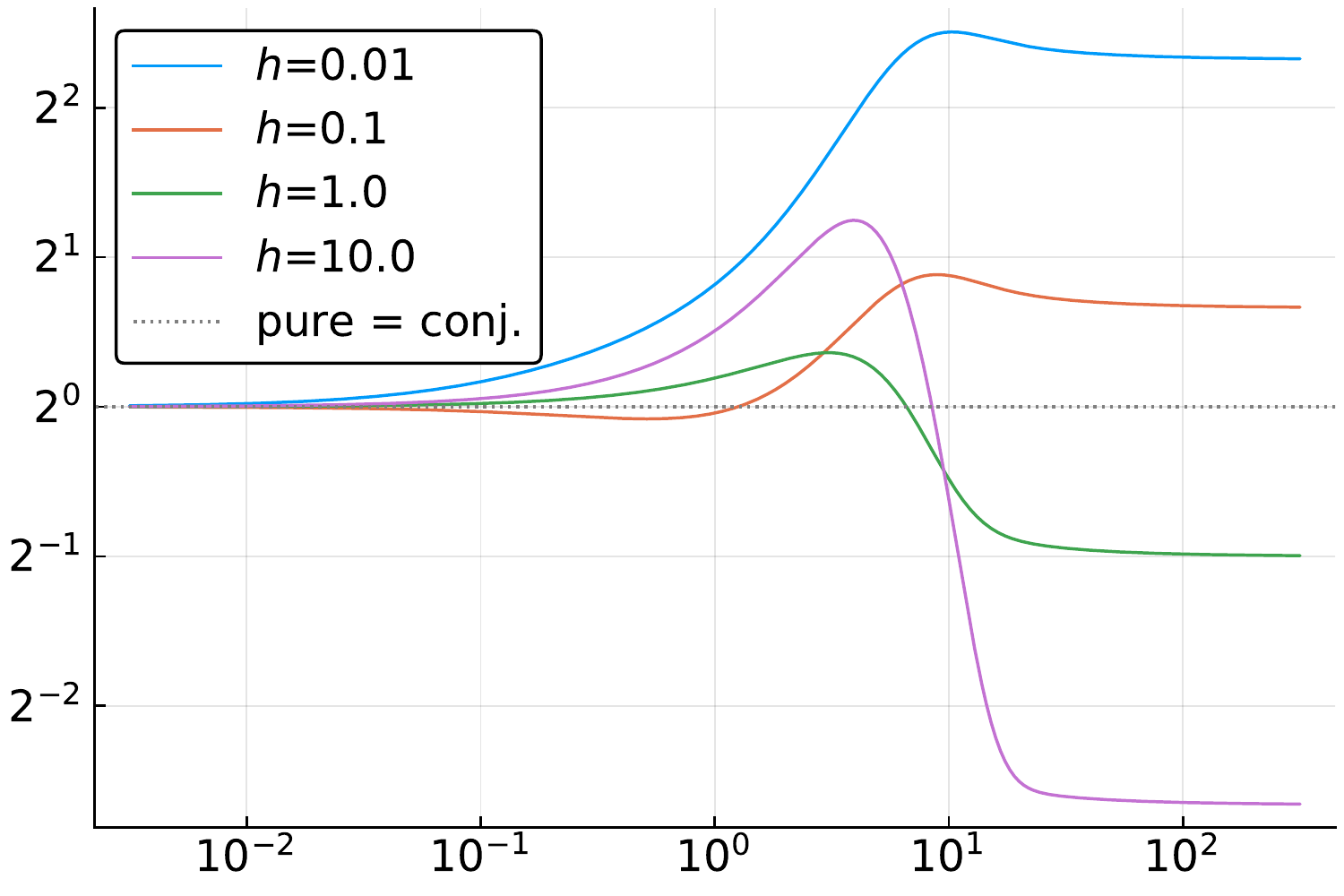}\tabularnewline
\begin{turn}{90}
1d, same width $\alpha$
\end{turn} & \includegraphics[width=0.4\columnwidth]{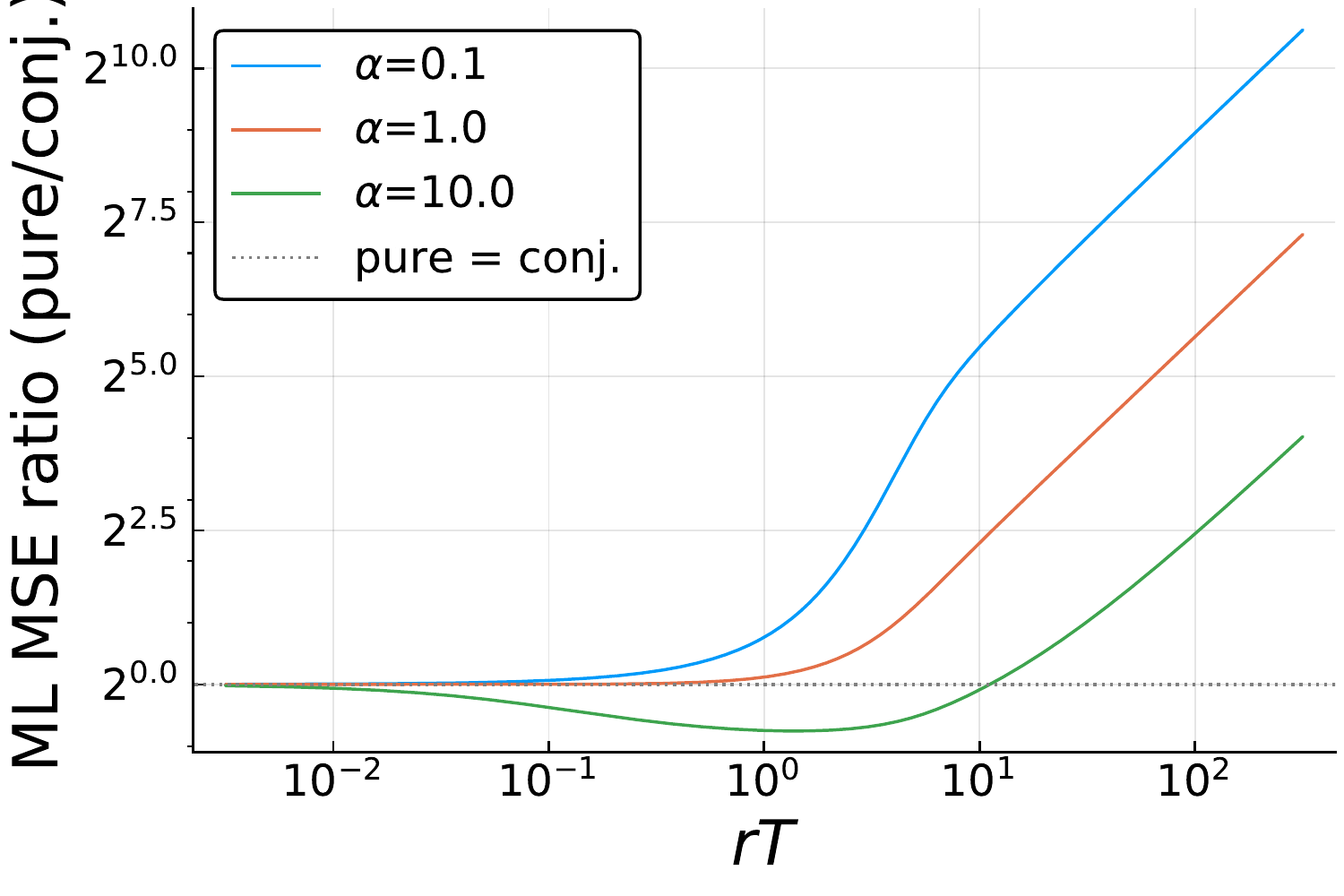} & \includegraphics[width=0.4\columnwidth]{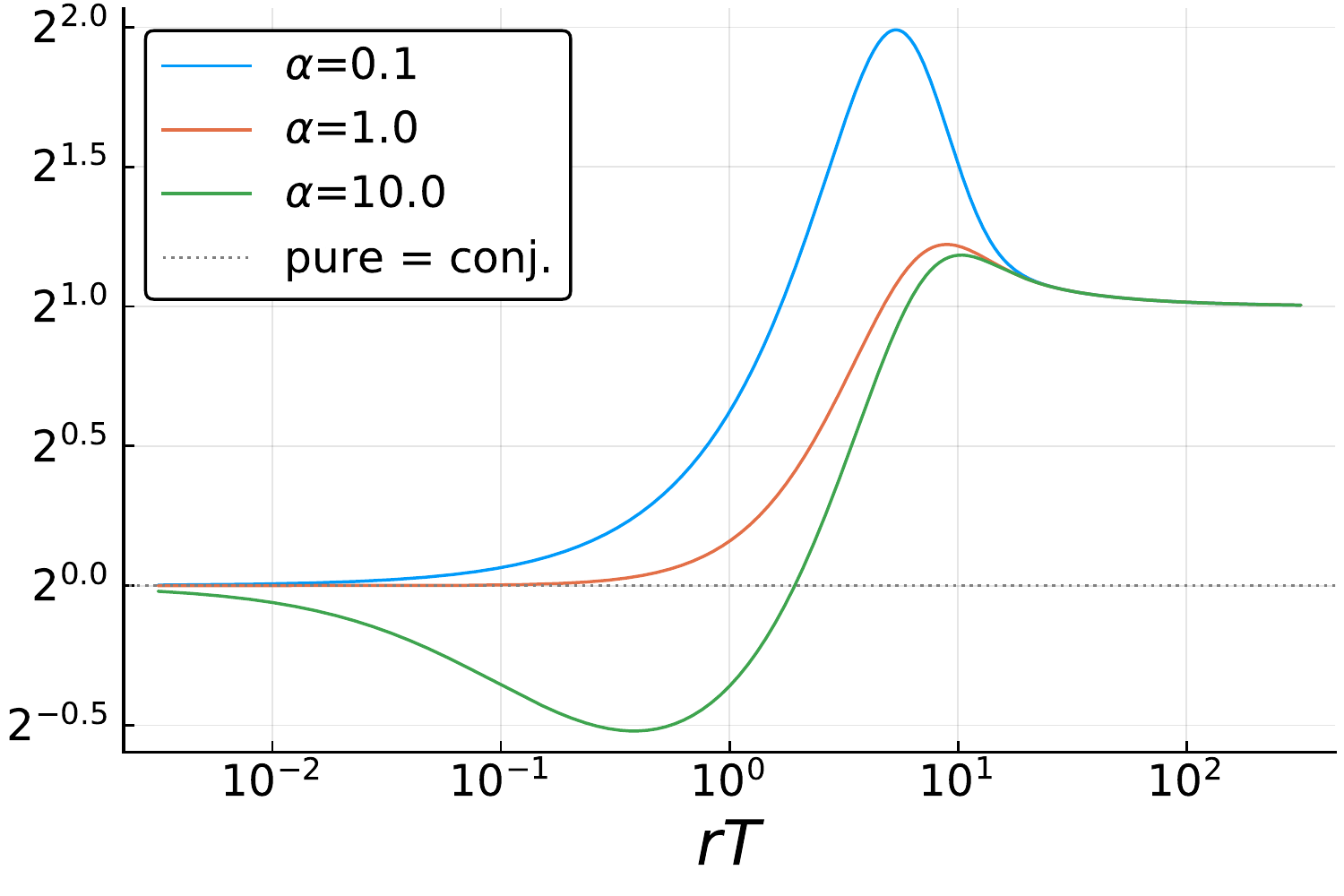}\tabularnewline
\end{tabular}
\par\end{centering}
\caption{ML MSE ratio of pure and conjunctive coding for various models.\emph{
}See text and Figure \ref{fig:pure-conj} for more details.}
\label{fig:pure-conj-ml}
\end{figure}
\par\end{center}

This analysis suggests that the different conclusion in \citep{Finkelstein2018}
is due to the cyclic nature of angle encoding. Experimentally observed
tuning curves in \citep{Finkelstein2018} are quite wide relative
to the range of possible angles, which might make a uniform population
over the unbounded space $\mathbb{R}^{m}$ an inadequate model. In
particular, the bounded space of angles allows for purely one-dimensional
encoding with non-zero tuning width and finite population rate. A
possible direction for future research is the analysis of MMSE-optimal
encoding in a uniform population encoding angles.
\end{document}